\documentclass[10pt,conference]{IEEEtran}

\renewcommand\thesection{\arabic{section}}
\renewcommand\thesubsection{\thesection.\arabic{subsection}}
\renewcommand\thesubsubsection{\thesubsection.\arabic{subsubsection}}

\renewcommand\thesectiondis{\arabic{section}}
\renewcommand\thesubsectiondis{\thesectiondis.\arabic{subsection}}
\renewcommand\thesubsubsectiondis{\thesubsectiondis.\arabic{subsubsection}}

\usepackage{verbatim} 
\usepackage{clrscode3e}
\usepackage[pdftex]{graphicx}
\usepackage{setspace}
\usepackage[cmex10]{amsmath}
\usepackage{amsmath, amsthm, amssymb}
\usepackage{algorithmic}
\usepackage{algorithm}
\usepackage{subfig}
\usepackage{caption}
\usepackage{paralist}
\usepackage{cases}
\usepackage{cite}

\usepackage{tikz}

\usepackage{etoolbox}

\usepackage{color}

\newcommand{\ceiling}[1]{\left\lceil{#1}\right\rceil}
\newcommand{\floor}[1]{\left\lfloor{#1}\right\rfloor}
\newcommand{\setof}[1]{\left\{{#1}\right\}}

\newcommand{\framework}[1]{$\mathbf{k^2U}$}
\newcommand{\frameworkkq}[1]{$\mathbf{k^2Q}$}

 \def\myendproof{{\ \vbox{\hrule\hbox{%
   \vrule height1.3ex\hskip0.8ex\vrule}\hrule }}\par}
 \renewenvironment{proof}{\noindent{\bf Proof. }}{\myendproof}
 \newenvironment{appProof}[1]{\noindent{\bf Proof of
     #1. }}{\myendproof\vskip 0.1in}

 \setboolean{ALC@noend}{true}
\providebool{techreport}
\setbool{techreport}{true}

\newcommand{\citetechreport}[1]{\ifbool{techreport}{}{ in the report \cite{DBLP:journals/corr/abs-1501.07084}}}

\newtheorem{theorem}{Theorem}
\newtheorem{lemma}{Lemma}
\newtheorem{corollary}{Corollary}
\newtheorem{example}{Example}

\newtheorem{definition}{Definition}

\graphicspath{{fig/multiframe/}{fig/arbitrary/}{fig/dag/}} 

\usetikzlibrary{%
  arrows,%
  shapes.misc,
  shapes.arrows,%
  chains,%
  matrix,%
  positioning,
  scopes,%
  decorations.pathmorphing,
  shadows%
}

\pgfdeclarelayer{background}
\pgfdeclarelayer{foreground}
\pgfsetlayers{background,main,foreground}
\tikzstyle{materia}=[draw, fill=white, text width=1.0em, text centered,
  minimum height=4em,drop shadow]
\tikzstyle{practica} = [materia, text width=18em, minimum width=8em,
align =left,
  minimum height=3em, rounded corners, drop shadow]
\tikzstyle{texto} = [above, text width=6em, text centered]
\tikzstyle{linepart} = [draw, thick, color=blue!50, -latex', dashed]
\tikzstyle{line} = [draw, line width = 2pt, color=blue!50, -latex']
\tikzstyle{ur}=[draw, text centered, minimum height=0.01em]

\newcommand{\practica}[2]{node (p#1) [practica]
  {\\{\footnotesize{#2}}}}

\ifbool{techreport}{
\pagestyle{plain}} 
{
\addtolength{\textheight}{14pt}
\pagestyle{empty}}

\title{\textbf{\textrm{k$^2$U}}: A General Framework from $k$-Point Effective Schedulability Analysis to Utilization-Based Tests}

\author{
    Jian-Jia Chen and Wen-Hung Huang\\
    Department of Informatics\\
    TU Dortmund University, Germany
    \and
    Cong Liu\\
    Department of Computer Science\\
    The University of Texas at Dallas
}
\vspace{3mm}

\begin{document}

\maketitle

\ifbool{techreport}{
\thispagestyle{plain}} 
{
\thispagestyle{empty}} 

\begin{abstract}

  To deal with a large variety of workloads in different application domains in real-time embedded systems, a number of expressive task models have been developed.  For each individual task model, researchers tend to develop different types of techniques for deriving schedulability tests with different computation complexity and performance.
In this paper, we present a general schedulability analysis framework, namely the \framework{} framework, that can be potentially applied to analyze a large set of real-time task models under any fixed-priority scheduling algorithm, on both uniprocessor and multiprocessor scheduling. The key to \framework{} is a $k$-point effective schedulability test, which can be viewed as a ``blackbox'' interface. For any task model, if a corresponding $k$-point effective schedulability test can be constructed, then a sufficient utilization-based test can be
automatically derived. We show the generality of \framework{} by applying it to different task models, which results in new and improved tests compared to the state-of-the-art.

\ifbool{techreport}{Analogously, a similar concept by testing
  only $k$ points with a different formulation has been studied by us in
  another framework, called \frameworkkq{}, which provides quadratic
  bounds or utilization bounds based on a different formulation of
  schedulability test.  With the quadratic and hyperbolic forms,
  \frameworkkq{} and \framework{} frameworks can be used to provide
  many quantitive features to be measured, like the total utilization
  bounds, speed-up factors, etc., not only for uniprocessor scheduling
  but also for multiprocessor scheduling.  These frameworks can be
  viewed as a ``blackbox'' interface for schedulability tests and
  response-time analysis.}{}
\end{abstract} 

\section{Introduction}
\label{sec:intro}

Given the emerging trend towards building complex cyber-physical
systems that often integrate external and physical devices, many
real-time and embedded systems are expected to handle a large variety
of workloads. Different formal real-time task models have been
developed to accurately represent these workloads with various
characteristics. Examples include the sporadic task
model~\cite{mok1983fundamental}, the multi-frame task
model~\cite{DBLP:dblp_journals/tse/MokC97}, the self-suspending task
model~\cite{suspension}, the directed-acyclic-graph (DAG) task model,
etc. Many of such formal models have been shown to be expressive
enough to accurately model real systems in practice. For example, the
DAG task model has been used to represent many computation-parallel
multimedia application systems and the self-suspending task model is
suitable to model workloads that may interact with I/O devices. For each of these task models, researchers tend to develop different types of techniques that result in schedulability tests with different computation complexity and performance (e.g., different utilization bounds).

\ifbool{techreport}{Over the years, real-time systems researchers have devoted a significant amount of time and efforts to efficiently analyze different task models. Many successful stories have been told. For many of the above-mentioned task models, efficient scheduling and schedulability analysis techniques have been developed (see \cite{DavisSurvey2011} for a recent survey). Unfortunately, for certain complex models such as the self-suspending task model, existing schedulability tests are rather pessimistic, particularly for the multiprocessor case (e.g., no utilization-based schedulability test exists for globally-scheduled multiprocessor self-suspending task systems). }{}

In this paper, we present \framework{}, a general schedulability analysis framework that is fundamentally based on a $k$-point effective schedulability test under fixed-priority scheduling. The key observation behind our proposed $k$-point test is the following. Traditional fixed-priority schedulability tests often have pseudo-polynomial-time (or even higher) complexity. For example, to verify the schedulability of a (constrained-deadline) task $\tau_k$ under fixed-priority scheduling in uniprocessor systems, the time-demand analysis (TDA) developed in \cite{DBLP:conf/rtss/LehoczkySD89}
can be adopted. That is, if
\begin{equation}
  \label{eq:exact-test-constrained-deadline}
\exists t \mbox{ with } 0 < t \leq D_k {\;\; and \;\;} C_k +
\sum_{\tau_i \in hp(\tau_k)} \ceiling{\frac{t}{T_i}}C_i \leq t,
\end{equation}
then task $\tau_k$ is schedulable under the fixed-priority scheduling algorithm, where $hp(\tau_k)$ is the set of the tasks with higher priority than $\tau_k$, $D_k$, $C_k$, and $T_i$ represent $\tau_k$'s relative deadline, worst-case execution time, and period, respectively. TDA incurs pseudo-polynomial-time complexity to check the time points that lie in $(0, D_k]$ for Eq.~\eqref{eq:exact-test-constrained-deadline}. 

To obtain sufficient schedulability tests under fixed priority
scheduling with reduced time complexity (e.g., polynomial-time), our
conceptual innovation is based on the observations by \emph{testing only a subset of such points} to
\emph{derive the minimum $C_k$ that cannot pass the
  schedulability tests}. This idea is
implemented in the \framework{} framework by providing a general
$k$-point effective schedulability test, which only needs to test $k$
points under \textit{any} fixed-priority scheduling when checking
schedulability of the task with the $k^{th}$ highest priority in the
system. This $k$-point effective schedulability test can be viewed as
a ``blackbox" interface that can result in sufficient
utilization-based tests.
 We show the generality of \framework{} by applying it to analyze several concrete example task models, including the constrained- and arbitrary-deadline sporadic task models, the multi-frame task model, the self-suspending task model, and the DAG task model. Note that \framework{} is not only applicable to uniprocessor systems, but also applicable to multiprocessor systems.

\noindent\textbf{Related Work.}
\ifbool{techreport}{An extensive amount of research has been conducted over the past forty years on verifying the schedulability of the classical sporadic task model in both uniprocessor and multiprocessor systems (see \cite{DavisSurvey2011} for a survey of such results). Much progress has also been made in recent years on analyzing more complex task models that are more expressive.}{}
There have been several results in the literature with respect to utilization-based, e.g., \cite{liu1973scheduling,HanTyan-RTSS97,journals/tc/LeeSP04,DBLP:conf/rtas/WuLZ05,kuo2003efficient,bini2003rate,RTSS14a}, \ifbool{techreport}{and non-utilization-based, e.g., \cite{DBLP:conf/rtss/ChakrabortyKT02,DBLP:conf/ecrts/FisherB05}, }{}
schedulability tests for the sporadic real-time task model and its generalizations in uniprocessor systems. 
\ifbool{techreport}{
The approaches in \cite{DBLP:conf/rtss/ChakrabortyKT02,DBLP:conf/ecrts/FisherB05} convert the stair function $\ceiling{\frac{t}{T_i}}$ in the time-demand analysis into a linear function if $t$ in Eq.~\eqref{eq:exact-test-constrained-deadline} is large enough. The methods in \cite{DBLP:conf/rtss/ChakrabortyKT02,DBLP:conf/ecrts/FisherB05} are completely different from this paper, in which the linear function of task $\tau_i$ starts after $\frac{t}{T_i} \geq k$, in which our method is based on $k$ points, defined individually by $\tau_1, \tau_2, \ldots, \tau_k$.

}{} 
Most of the existing utilization-based schedulability analyses focus on the total utilization bound. That is, if the total utilization of the task system is no more than the derived bound, the task system is schedulable by the scheduling policy. For example, the total utilization bounds derived in \cite{liu1973scheduling,HanTyan-RTSS97,DBLP:dblp_journals/tc/BurchardLOS95} are mainly for rate-monotonic (RM) scheduling, in which the results in \cite{HanTyan-RTSS97} can be extended for arbitrary fixed-priority scheduling. Kuo et al. \cite{kuo2003efficient} further improve the total utilization bound by using the notion of divisibility. Lee et al. \cite{journals/tc/LeeSP04} use linear programming formulations for calculating total utilization bounds
when the period of a task can be selected. Moreover, Wu et al. \cite{DBLP:conf/rtas/WuLZ05} adopt the Network Calculus to analyze the total utilization bounds of several task models.

The novelty of \framework{} comes from a different perspective
from these approaches
\cite{liu1973scheduling,HanTyan-RTSS97,journals/tc/LeeSP04,DBLP:conf/rtas/WuLZ05,kuo2003efficient}. We
do not specifically seek for the total utilization bound. Instead, we
look for the critical value in the specified sufficient schedulability
test while verifying the schedulability of task $\tau_k$. 
A natural schedulability condition to express
the schedulability of task $\tau_k$ is a hyperbolic bound, (to be
shown in Lemma \ref{lemma:framework-constrained}), whereas the
corresponding total utilization bound can be obtained (in Lemmas
\ref{lemma:framework-totalU-constrained} and
\ref{lemma:framework-totalU-exclusive}). 

The hyperbolic forms are the centric features in \framework{} analysis, in which the test by Bini et al. \cite{bini2003rate}  for sporadic real-time tasks and our recent result in \cite{RTSS14a} for bursty-interference analysis are both special cases and simple implications from the \framework{} framework. With the hyperbolic forms, we are then able to provide many interesting observations with respect to the required quantitive features to be measured, like the total utilization bounds, speed-up factors, etc., not only for uniprocessor scheduling but also for multiprocessor scheduling.
For more details, we will provide further explanations at the end of Sec. \ref{sec:framework} after the framework is presented. For the studied task models to demonstrate the applicability of \framework{}, we will summarize some of the latest results on these task models in their corresponding sections.

\noindent\textbf{Contributions.} In this paper, we present a general
schedulability analysis framework, \framework{}, that can be applied
to analyze a number of complex real-time task models, on both
uniprocessors and multiprocessors.
For any
task model, if a corresponding $k$-point effective schedulability test
can be constructed, then a sufficient utilization-based test can be
derived by the \framework{} framework. We show the generality of \framework{} by applying it to several task models, in which the results are better or more general compared to the state-of-the-art:

\begin{enumerate}
\item For uniprocessor constrained-deadline sporadic task systems, the speed-up factor of our obtained schedulability test is 1.76322. This value is the same as the lower bound and upper bound of deadline-monotonic (DM) scheduling shown by Davis et al.~\cite{DBLP:journals/rts/DavisRBB09}. Our result is thus stronger (and requires a much simpler proof), as we show that the same factor holds for a polynomial-time schedulability test (not just the DM scheduler). For uniprocessor arbitrary-deadline sporadic task systems, our obtained utilization-based test works for any fixed-priority scheduling with arbitrary priority-ordering assignment.
\item For multiprocessor DAG task systems under global rate-monotonic (RM) scheduling, the capacity-augmentation
  factor, as defined in ~\cite{Li:ECRTS14} and Sec.~\ref{sec:multiprocessor} in this paper, of our obtained test is
  3.62143. This result is better than the best existing result, which
  is 3.73, given by Li et al.~\cite{Li:ECRTS14}. Our result is also applicable for conditional sporadic DAG task systems \cite{DBLP:conf/ecrts/BaruahBM15}.

\item For multiprocessor self-suspending task systems, we obtain the \textit{first} utilization-based test for global RM.
\item For uniprocessor multi-frame task systems, our obtained utilization bound is superior to the results by Mok and Chen~\cite{DBLP:dblp_journals/tse/MokC97} analytically and Lu et al. \cite{lu2007new} in our simulations. \ifbool{techreport}{The analysis is in Appendix C, whereas the evaluation result is in Appendix D.}{Due to space limitations, the analysis is in Appendix C in the report \cite{DBLP:journals/corr/abs-1501.07084}, whereas the evaluation result is in Appendix D in  \cite{DBLP:journals/corr/abs-1501.07084}.}
\end{enumerate}

Note that the emphasis of this paper is not to show that the resulting tests for different task models by applying the \framework{} framework are better than existing work. Rather, we want to show that the \framework{} framework is general, easy to use, and has relatively low time complexity, but is still able to generate good tests. By demonstrating the applicability of the \framework{} framework to several task models, we believe that this  framework has great potential in analyzing many other complex real-time task models, where the existing analysis approaches are insufficient or cumbersome. To the best of our knowledge, together with \frameworkkq{} to be explained later, these are the first general schedulability analysis frameworks that can be potentially applied to analyze a large set of real-time task models under any fixed-priority scheduling algorithm in both uniprocessor and multiprocessor systems.

\noindent{\bf Comparison to \frameworkkq{}:} 
The concept of testing $k$ points only is also the key in
another framework designed by us, called \frameworkkq{}
\cite{DBLP:journals/corr/abs-k2q}.  Even though \frameworkkq{} and \framework{} share the
same idea by testing and evaluating only $k$ points, they are based on
completely different criteria for testing.  In \framework{}, all the testings and
formulations are based on \emph{only the higher-priority task
  utilizations}. In \frameworkkq{}, the testings are based \emph{not
  only on the higher-priority task utilizations, but also on the
  higher-priority task execution times}.  The above difference
in the formulations results in completely different properties and
mathematical closed-forms.  The
natural schedulability condition of \framework{} is a \emph{hyperbolic form} for testing
the schedulability, whereas the natural schedulability condition of \frameworkkq{} is a
\emph{quadratic form} for testing the schedulability or the response
time of a task.

\emph{If one framework were dominated by another or these two frameworks
were just with minor difference in mathematical formulations, it
wouldn't be necessary to separate and present
them as two different frameworks.} Both frameworks are in fact needed
and have to be applied for different cases.  Due to space limitation,
we can only shortly explain their differences,
advantages, and disadvantages in this paper.  For completeness, another
document has been prepared in
\cite{DBLP:journals/corr/framework-compare} to present the similarity,
the difference and the characteristics of these two frameworks in details.

Since the formulation of \framework{} is more restrictive than
\frameworkkq{}, its applicability is limited by the possibility to
formulate the tests purely by using higher-priority task utilizations
without referring to their execution times. There are cases, in which
formulating the higher-priority interference by using only task
utilizations for \framework{} is troublesome. For such cases, further
introducing the upper bound of the execution time by using
\frameworkkq{} is more precise. 
\ifbool{techreport}{The above cases can be found in (1)
the schedulability tests for arbitrary-deadline sporadic task systems
in uniprocessor scheduling, (2) multiprocessor global fixed-priority
scheduling when adopting the forced-forwarding schedulability test,
etc.\footnote{ c.f. Sec. 5/6 in \cite{DBLP:journals/corr/abs-k2q} for
  the detailed proofs and Sec. 5/6 in
  \cite{DBLP:journals/corr/framework-compare} for the performance
  comparisons.}}{} 
In general, if we can formulate the schedulability
tests into the \framework{} framework by purely using higher-priority
task utilizations, it is also usually possible to formulate it into
the \frameworkkq{} framework by further introducing the task execution
times. In such cases, the same pseudo-polynomial-time (or exponential
time) test is used, and the
utilization bound or speed-up factor analysis 
derived from the \framework{} framework is, in
general, tighter and better.


\ifbool{techreport}{
In a nutshell, with respect to quantitive metrics, like utilization bounds or speedup factor analysis, \framework{} is more precise, whereas \frameworkkq{} is
more general. If the exact (or very precise) schedulability test can
be constructed and the test can be converted into \framework{},
e.g., uniprocessor scheduling for
constrained-deadline task sets, then, adopting \framework{} \emph{may} lead to
tight results. By adopting \frameworkkq{}, we may be able to start
from a more complicated test with exponential-time complexity and
convert it to a linear-time approximation with better results than
\framework{}. Although \framework{} is more restrictive than
\frameworkkq{}, both of them are general enough to cover a range of
applications, ranging from uniprocessor systems to
multiprocessor systems.
For more information and comparisons, please refer to \cite{DBLP:journals/corr/framework-compare}.
}
{
}

\section{Sporadic Task and Scheduling Models}
\label{sec:model}

A sporadic task $\tau_i$ is released repeatedly, with each such 
invocation called a job. The $j^{th}$ job of $\tau_i$, denoted
$\tau_{i,j}$, is released at time $r_{i,j}$ and has an absolute deadline at time
$d_{i,j}$. Each job of any task $\tau_i$ is assumed to have
execution time $C_i$. Here in this paper, whenever we refer to
the execution time of a job, we mean for the worst-case execution time
of the job since all the analyses we use are safe by only considering the worst-case execution time.  Successive jobs of the same task are
required to execute in sequence. Associated with each task $\tau_i$
are a period $T_i$, which specifies the minimum time between two
consecutive job releases of $\tau_i$, and a deadline $D_i$, which
specifies the relative deadline of each such job, i.e.,
$d_{i,j}=r_{i,j}+D_i$. The utilization of a task $\tau_i$ is defined
as $U_i=C_i/T_i$.

A sporadic task system $\tau$ is said to be an implicit-deadline
system if $D_i = T_i$ holds for each $\tau_i$. A sporadic task system
$\tau$ is said to be a constrained-deadline system if $D_i \leq T_i$
holds for each $\tau_i$.  Otherwise, such a sporadic task system
$\tau$ is an arbitrary-deadline system.

A task is said \emph{schedulable} by a scheduling policy if all of its
jobs can finish before their absolute deadlines.  A task system is
said \emph{schedulable} by a scheduling policy if all the tasks in the
task system are schedulable. A \emph{schedulability test} 
expresses sufficient conditions to ensure the feasibility of the
resulting schedule by a scheduling policy.

Throughout the paper, we will focus on fixed-priority preemptive
scheduling. That is, each task is associated with a priority level.
For a uniprocessor system, i.e., except
Sec.~\ref{sec:multiprocessor}, the scheduler always dispatches the
job with the highest priority in the ready queue to be executed. For a
uniprocessor system, it has been shown that RM
scheduling is an optimal fixed-priority scheduling policy for
implicit-deadline systems \cite{liu1973scheduling} and
DM scheduling is an optimal fixed-priority
scheduling policy for constrained-deadline
systems\cite{journals/pe/LeungW82}.

To verify the schedulability of task $\tau_k$ under
fixed-priority scheduling in uniprocessor systems, the time-demand
analysis developed in \cite{DBLP:conf/rtss/LehoczkySD89}
can be adopted, as discussed earlier. That is, if Eq.~\eqref{eq:exact-test-constrained-deadline} holds,
 then task $\tau_k$ is schedulable under the fixed-priority scheduling
algorithm.
For the simplicity of presentation,
we will demonstrate how the framework works using ordinary sporadic real-time task systems 
in Sec.~\ref{sec:framework} and Sec.~\ref{sec:application-FP}. We
will demonstrate more applications with respect to multi-frame tasks
\cite{DBLP:dblp_journals/tse/MokC97} in Appendix C\citetechreport{} 
and with respect to multiprocessor scheduling in
Sec.~\ref{sec:multiprocessor}.

\vspace{-2mm}
\section{Analysis Flow}
\label{sec:flow}

The proposed \framework{} framework only tests the
schedulability of a specific task $\tau_k$, under the assumption
that the higher-priority tasks are already verified to be schedulable
by the given scheduling policy. Therefore, this framework has to be
applied for each of the given tasks. A task system is schedulable by
the given scheduling policy only when all the tasks in the system can
be verified to meet their deadlines.  The results can be extended to
test the schedulability of a task system in linear time complexity or to allow on-line admission control in constant time complexity if the schedulability
condition (or with some more pessimistic simplifications) is
monotonic. Such extensions are provided in Appendix A for some cases.

Therefore, for the rest of this paper, we implicitly assume that all
the higher-priority tasks are already verified to be schedulable by
the scheduling policy.  We will only present the schedulability test
of a certain task $\tau_k$, that is being analyzed, under the
above assumption. For notational brevity, we implicitly assume that there
are $k-1$ tasks, say $\tau_1, \tau_2, \ldots, \tau_{k-1}$ with
higher-priority than task $\tau_k$. Moreover, we only consider the cases when $k \geq 2$, since $k=1$ is pretty trivial.

\section{\framework{} Framework}
\label{sec:framework}

This section presents the basic properties of the \framework{}
framework for testing the schedulability of task $\tau_k$ in a given
set of real-time tasks (depending on the specific models given in
each application as shown later in this paper). We will first provide examples
to explain and define the \emph{$k$-point effective schedulability
  test}. Then, we will provide the fundamental properties of the
corresponding utilization-based tests. Throughout this section, we
will implicitly use sporadic task systems defined in
Sec.~\ref{sec:model} to simplify the presentation. The concrete
applications will be presented in Secs.~\ref{sec:application-FP} - \ref{sec:multiprocessor}.

The $k$-point effective schedulability test is a sufficient
schedulability test that verifies only $k$ time points, defined by the
$k-1$ higher-priority tasks and task $\tau_k$. For example, instead of
testing all the possible $t$ in the range of $0$ and $D_k$ in
Eq.~\eqref{eq:exact-test-constrained-deadline}, we can simply test
only $k$ points. It may seem to be very pessimistic to only test
$k$ points. However, if these $k$ points are \textit{effective},\footnote{As to be clearly illustrated later, the $k$ points can be considered effective if they can define certain extreme cases of task parameters. For example, the ``difficult-to-schedule'' concept first introduced by Liu and Layland~\cite{liu1973scheduling} defines $k$ effective points that are used in Example \ref{example-1}. In their case~\cite{liu1973scheduling}, the selected set of the $k$ points was ``very'' effective because the tested task $\tau_k$ becomes unschedulable if $C_k$ is increased by an arbitrarily small value $\varepsilon$.} the resulting
schedulability test may be already good. We now demonstrate two
examples.

\begin{example}
\label{example-1}
  {\bf Implicit-deadline task systems}: Suppose that the tasks are indexed by
  the periods, i.e., $T_1 \leq \cdots \leq
  T_k$. When $T_k \leq 2T_1$, task $\tau_k$ is schedulable by RM if
  there exists $j \in \setof{1,2,\ldots,k}$ where
\begin{equation}
  \label{eq:precodition-schedulability-sporadic}\small
  C_k + \sum_{i=1}^{k-1} C_i + \sum_{i=1}^{j-1} C_i = C_k + \sum_{i=1}^{k-1} T_i U_i + \sum_{i=1}^{j-1} T_i U_i \leq T_j.   
\end{equation}\normalsize  \endproof
\end{example}
That is, in the above example, it is sufficient to only test $T_1,
T_2, \ldots, T_k$. The case defined in the above example is utilized
by Liu and Layland \cite{liu1973scheduling} for deriving the least
utilization upper bound $69.3\%$ for RM scheduling. We can
make the above example more generalized as follows:

\begin{example}
\label{example-2}
{\bf Implicit-deadline task systems with given ratios of periods}:
Suppose that $f T_i \leq T_k$ for a given integer $f$ with $f \geq 1$
for any higher-priority task $\tau_i$, for all $i=1,2,\ldots,k-1$.
Let $t_i$ be 
$\floor{\frac{T_k}{T_i}}T_i$. Suppose that the $k-1$ higher
priority tasks are indexed such that $t_1 \leq t_2 \leq \cdots \leq
t_{k-1} \leq t_k$, where $t_k$ is defined as $T_k$.  Task $\tau_k$ is
schedulable under RM if there exists $j$ with $1 \leq j \leq k$ such that
\begin{equation}
  \label{eq:precodition-schedulability-sporadic-f}
 C_k + \sum_{i=1}^{k-1} t_i U_i + \sum_{i=1}^{j-1} C_i \leq  C_k +
 \sum_{i=1}^{k-1} t_i U_i + \sum_{i=1}^{j-1} \frac{1}{f} t_i U_i\leq t_j,  
\end{equation}
where the first inequality in
Eq.~\eqref{eq:precodition-schedulability-sporadic-f} is due to the
fact $C_i = T_i U_i \leq \frac{1}{f} t_i U_i$.  That is, in the above example, it is sufficient to only test
$\floor{\frac{T_k}{T_1}}T_1, \floor{\frac{T_k}{T_2}}T_2, \ldots, 
\floor{\frac{T_k}{T_{k-1}}}T_{k-1}, T_k$.
  \endproof
\end{example}

With the above examples, for a given set $\setof{t_1, t_2, \ldots
  t_k}$, we now define the $k$-point effective schedulability test as follows:
\begin{definition}
  \label{def:kpoints}
  A $k$-point effective schedulability test is a sufficient
  schedulability test of a fixed-priority scheduling policy, that verifies the existence of $t_j \in \setof{t_1, t_2, \ldots t_k}$ with $0 < t_1 \leq t_2 \leq \cdots \leq t_k$ such that \begin{equation}
    \label{eq:precodition-schedulability}
    C_k + \sum_{i=1}^{k-1} \alpha_i t_i U_i + \sum_{i=1}^{j-1} \beta_i t_i U_i \leq t_j,
  \end{equation}
  where $C_k > 0$, $\alpha_i > 0$, $U_i > 0$, and $\beta_i >0$ are dependent upon the setting
  of the task models and task $\tau_i$.
\end{definition}
 For Example~\ref{example-1}, the effective values in
$\setof{t_1, t_2, \ldots t_k}$ are $T_1, T_2, \ldots, T_k$, and
$\alpha_i=\beta_i=1$ for each task $\tau_i$. For
Example~\ref{example-2}, the effective values in $\setof{t_1,
  t_2, \ldots t_k}$ are with $\alpha_i=1$ and $\beta_i \leq
\frac{1}{f}$ for each task $\tau_i$.


Moreover, we only consider non-trivial cases,
in which $C_k >0$, $t_k > 0$, $0 < \alpha_i$, $0 < \beta_i$, and $0 < U_i \leq 1$ for $i=1,2,\ldots,k-1$.
The definition of the $k$-point last-release schedulability test 
$C_k + \sum_{i=1}^{k-1} \alpha_i t_i U_i + \sum_{i=1}^{j-1}
\beta_i t_i U_i \leq t_j$ 
in Definition \ref{def:kpoints} only slightly differs from the
test 
$C_k +
\sum_{i=1}^{k-1} \alpha_i t_i U_i + \sum_{i=1}^{j-1} \beta_i C_i \leq
t_j$
in the \frameworkkq{} framework \cite{DBLP:journals/corr/abs-k2q}.
However, since the tests are different, they are used for different
situations.

With these $k$ points, we are able to define the corresponding
coefficients $\alpha_i$ and $\beta_i$ in the $k$-point effective
schedulability test of a scheduling algorithm. The elegance of the
\framework{} framework is to use only the parameters $\alpha_i$ and
$\beta_i$ to analyze whether task $\tau_k$ can pass the
schedulability test. Therefore, the \framework{} framework provides
corresponding utilization-based tests \emph{automatically} if the $k$-point effective
schedulability test and the corresponding parameters $\alpha_i$ and
$\beta_i$ can be defined, which will be further demonstrated in
the following sections with several applications.

We are going to present the properties
resulting from the $k$-point effective schedulability test under given
$\alpha_i$ and $\beta_i$.  In the following lemmas, we are going to
seek the extreme cases for these $k$ testing points under the given
setting of utilizations and the defined coefficients $\alpha_i$ and
$\beta_i$.  To make the notations clear, these extreme testing points
are denoted as $t_i^*$ for the rest of this paper.  The procedure will
derive $k-1$ extreme testing points, denoted as $t_1^*, t_2^*, \ldots,
t_{k-1}^*$, whereas $t_k^*$ is defined as $t_k$ plus a slack variable (to be defined in the proof of Lemma~\ref{lemma:framework-constrained}) for notational
brevity.
Lemmas~\ref{lemma:framework-constrained} to \ref{lemma:framework-totalU-exclusive} are useful to analyze the utilization bound, the hyperbolic bound, etc., for given scheduling strategies, when 
$\alpha$ and $\beta$ can be easily defined based on the scheduling policy, with
$0 < t_k$ and
$0 < \alpha_i \leq \alpha$, and $0 < \beta_i \leq \beta$ for any
$i=1,2,\ldots,k-1$.

\begin{lemma}
\label{lemma:framework-constrained}
For a given $k$-point effective schedulability test of a scheduling 
algorithm, defined in
Definition~\ref{def:kpoints},
in which $0 < t_k$ and $0 < \alpha_i \leq \alpha$, and $0 < \beta_i \leq \beta$ for any
$i=1,2,\ldots,k-1$, task $\tau_k$ is schedulable by the scheduling
algorithm if the following condition holds
\begin{equation}
\label{eq:schedulability-constrained}
\frac{C_k}{t_k} \leq \frac{\frac{\alpha}{\beta}+1}{\prod_{j=1}^{k-1} (\beta U_j + 1)} - \frac{\alpha}{\beta}.
\end{equation}
\end{lemma}
\begin{proof}
  If $\frac{\frac{\alpha}{\beta}+1}{\prod_{j=1}^{k-1} (\beta U_j + 1)}
  - \frac{\alpha}{\beta} \leq 0$, the condition in
  Eq.~\eqref{eq:schedulability-constrained} never holds since $C_k > 0$, and the
  statement is vacuously true.  We focus on the  case $\frac{\frac{\alpha}{\beta}+1}{\prod_{j=1}^{k-1}
    (\beta U_j + 1)} - \frac{\alpha}{\beta} > 0$.

  We prove this lemma by showing that the condition in
  Eq.~\eqref{eq:schedulability-constrained} leads to the satisfactions of the
  schedulability condition listed in
  Eq.~(\ref{eq:precodition-schedulability}) by using contrapositive.
  By taking the negation of the schedulability condition in
  Eq.~(\ref{eq:precodition-schedulability}), we know that if task
  $\tau_k$ is \emph{not schedulable} by the scheduling policy, then
  for each $j=1,2,\ldots, k$
  \begin{equation}
    \label{eq:lp-init-constraints}
    C_k + \alpha \sum_{i=1}^{k-1} t_i U_i + \beta\sum_{i=1}^{j-1} t_i U_i \geq 
    C_k + \sum_{i=1}^{k-1} \alpha_i t_i U_i + \sum_{i=1}^{j-1} \beta_i t_i U_i > t_j, 
  \end{equation}
  due to $0 < \alpha_i \leq \alpha$, and $0 < \beta_i \leq \beta$ for any
$i=1,2,\ldots,k-1$
  To enforce the condition in
  Eq.~\eqref{eq:lp-init-constraints}, we are going to
  show that $C_k$ must have some lower bound. Therefore, if $C_k$ is
  no more than this lower bound, then task $\tau_k$ is schedulable by the
  scheduling policy. 

  For the rest of the proof, we replace $>$ with $\geq$ in
  Eq.~\eqref{eq:lp-init-constraints}, as the infimum and the minimum
  are the same when presenting the inequality with $\geq$.  Moreover,
  we also relax the problem by replacing the constraint $t_{j+1} \geq
  t_j$ with $t_j \geq 0$ for $j=1,2,\ldots,k-1$. Therefore, the
  unschedulability for satisfying Eq.~\eqref{eq:lp-init-constraints}
  implies that $C_k > C_k^*$, where $C_k^*$ is
  defined in the following optimization problem:
  \begin{subequations}\label{eq:lp-init}
\small  \begin{align}
    \mbox{min\;\;} & C_k^* \label{eq:precodition-schedulability-objective}\\
    \mbox{s.t.\;\;} &     C_k^* + \sum_{i=1}^{k-1} \alpha t_i^* U_i + \sum_{i=1}^{j-1} \beta t_i^* U_i \geq t_j^* &\forall j=1,\ldots, k-1,     \label{eq:precodition-schedulability-negation-0}\\
    &    t_j^* \geq 0\;&\forall j=1,\ldots, k-1,     \label{eq:precodition-schedulability-negation-1}\\
     &     C_k^* + \sum_{i=1}^{k-1} (\alpha +\beta) t_i^* U_i  \geq t_k,     \label{eq:precodition-schedulability-negation-2}
  \end{align}    
  \end{subequations}
  where $t^*_1, t^*_2, \ldots, t^*_{k-1}$ and $C_k^*$ are variables; and $\alpha$, $\beta$ are constants. 

  Let $s \geq 0$ be a slack variable such that $C_k^* +
  \sum_{i=1}^{k-1} (\alpha +\beta) t_i^* U_i = t_k+s$.  Therefore, $C_k^* = 
   t_k+s -\sum_{i=1}^{k-1} (\alpha +\beta) t_i^* U_i$. By replacing $C_k^*$ in Eq.~\eqref{eq:precodition-schedulability-negation-0}, we have
 \begin{align}
   & t_k+s - (\sum_{i=1}^{k-1} \alpha t_i^*U_i + \sum_{i=1}^{k-1} \beta t_i^*U_i) + \sum_{i=1}^{k-1} \alpha t_i^*U_i + \sum_{i=1}^{j-1} \beta t_i^*U_i\nonumber\\
    =\; & t_k+s- \sum_{i=j}^{k-1} \beta t_i^* U_i 
    \geq t_j^*, \;\;\;\forall j=1,\ldots, k-1. \label{eq:precodition-schedulability-negation}
  \end{align}
  For notational brevity, let $t_k^*$ be $t_k+s$.
 Therefore, the linear programming in Eq.~\eqref{eq:lp-init} can be rewritten as follows:
\begin{subequations}
    \label{eq:lp-framework-constrained}
  \begin{align}
    \mbox{min } &  
    t_k^* -\sum_{i=1}^{k-1} (\alpha +\beta) t_i^* U_i\\
    \mbox{s.t.} \;\;&
t_k^*-\beta \sum_{i=j}^{k-1} t_i^* U_i \geq t_j^*,
      & \forall 1 \leq j \leq k - 1 \label{eq:lp-framework-constrained-constraints}\\
&t_j^* \geq 0 
      & \forall 1 \leq j \leq k - 1 \label{eq:lp-framework-constrained-boundaryconstraints}\\
&t_k^* \geq t_k \label{eq:lp-framework-constrained-slack-constraints}
  \end{align}    
  \end{subequations}  

  The remaining proof is to solve the above linear programming to obtain the minimum $C_k^*$ if $\frac{\frac{\alpha}{\beta}+1}{\prod_{j=1}^{k-1} (\beta U_j + 1)} - \frac{\alpha}{\beta} > 0$. 
Our proof strategy is to solve the linear programming analytically as a function of $t_k^*$. This can be imagined as if $t_k^*$ is given. At the end, we will prove the optimality by considering all possible $t_k^* \geq t_k$.
  This involves three steps:
  \begin{itemize}
  \item Step 1: we analyze certain properties of optimal solutions based on the extreme point theorem for linear programming \cite{luenberger2008linear} under the assumption that $t_k^*$ is given as a constant, i.e., $s$ is known.
  \item Step 2: we present a specific solution in an \emph{extreme point}, as a function of $t_k^*$.
  \item Step 3: we prove that the above extreme point solution gives the minimum $C_k^*$  if $\frac{\frac{\alpha}{\beta}+1}{\prod_{j=1}^{k-1} (\beta U_j + 1)} - \frac{\alpha}{\beta} > 0$.
  \end{itemize}

  
  {\bf [Step 1:]} In this step, we assume that
  $s$ is given,  i.e., $t_k^*$ is specified as a constant. The original linear
  programming in Eq.~(\ref{eq:lp-framework-constrained}) has $k$
  variables and $2(k-1)+1$ constraints. After specifying the value
  $t_k^*$ as a given constant, the new linear programming without the
  constraint in
  Eq.~\eqref{eq:lp-framework-constrained-slack-constraints} has only
  $k-1$ variables and $2(k-1)$ constraints.  Thus, according to the
  extreme point theorem for linear programming
  \cite{luenberger2008linear}, the linear constraints form a
  polyhedron of feasible solutions. The extreme point theorem states
  that either there is no feasible solution or one of the extreme
  points in the polyhedron is an optimal solution when the objective
  of the linear programming is finite. To satisfy
  Eqs.~\eqref{eq:lp-framework-constrained-constraints}
  and~\eqref{eq:lp-framework-constrained-boundaryconstraints}, we know
  that $t_j^* \leq t_k^*$ for $j=1,2,\ldots,k-1$, due to $t_i^* \geq
  0$, $0 < \beta$, $0 < \alpha$ and $U_i > 0$
  for $i=1,2, \ldots,k-1$. As a result, the objective of the above
  linear programming is $t_k^* -\sum_{i=1}^{k-1} (\alpha +\beta) t_i^* U_i \geq t_k^* -\sum_{i=1}^{k-1} (\alpha +\beta) t_k^* U_i$, which is finite (as a function of $t_k^*$, $\alpha$, and $\beta$) under the assumption $0 < \beta$, $0 < \alpha$ and $0 < \sum_{i=1}^{k-1}U_i \leq 1$.

  According to the extreme point theorem, one of the extreme points is
  the optimal solution of Eq.~\eqref{eq:lp-framework-constrained} for a given $t_k^*$.
  There are $k-1$ variables with $2(k-1)$ constraints in
  Eq.~\eqref{eq:lp-framework-constrained} for a given $t_k^*$. An extreme point must have
  at least $k-1$ \emph{active} constraints in
  Eqs.~\eqref{eq:lp-framework-constrained-constraints} and \eqref{eq:lp-framework-constrained-boundaryconstraints},  in which
  their $\geq$ are set to equality $=$.

We now prove that an extreme point solution is feasible for Eq.~\eqref{eq:lp-framework-constrained} by setting $t_j^*$ either to $0$ or to $t_k^*-\beta \sum_{i=j}^{k-1} t_i^* U_i = t_j^* > 0$ for $j=1,2,\ldots,k-1$.
Suppose for contradiction that there exists a task $\tau_h$ with $0 = t_h^* = t_k^*-\beta \sum_{i=h}^{k-1} t_i^* U_i$.
Let $\omega$ be the 
index of the next task after $\tau_h$ with $t_{\omega}^* > 0$ in this extreme point solution, i.e., $t_{h}^* = t_{h+1}^* = \cdots= t_{\omega-1}^*=0$. If $t_i^*=0$ for $h \leq i \leq k-1$, then $\omega$ is set to $k$ and $t_{\omega}^*$ is $t_k^*$. Therefore, the conditions $t_k^*-\beta \sum_{i=h}^{k-1} t_i^* U_i = t_h^*=0$ and  $t_{h+1}^* = \cdots= t_{\omega-1}^*=0$ imply the contradiction that $0=t_h^* = t_k^*-\beta \sum_{i=h}^{k-1} t_i^* U_i = t_k^*-\beta \sum_{i=\omega}^{k-1} t_i^* U_i \geq t_{\omega}^* > 0$ when $\omega \leq k-1$ or $0 = t_h^* =t_k^*-\sum_{i=h}^{k-1} t_i^* U_i  =t_k^*>0$ when $\omega=k$. 

By the above analysis and the pigeonhole principle, a feasible extreme point solution of Eq.~\eqref{eq:lp-framework-constrained} can be represented by two sets ${\bf T}_1$ and ${\bf T}_2$ of the $k-1$ higher-priority tasks, in which $t_j^* = 0$ if $\tau_j$ is in ${\bf T}_1$ and $t_k^*-\beta \sum_{i=j}^{k-1} t_i^* U_i = t_j^* > 0$ if task $\tau_j$ is in ${\bf T}_2$. With the above discussions, we have ${\bf T}_1\cup{\bf T}_2=\setof{\tau_1, \tau_2, \ldots, \tau_{k-1}}$ and ${\bf T}_1\cap {\bf T}_2 = \emptyset$. 

{\bf [Step 2:]}
For a given $t_k^*$, one special extreme point solution is to put $t_k^* - \beta \sum_{i=j}^{k-1} t_i^* U_i = t_j^*$ for every $j=1,2,\ldots,k-1$, i.e., ${\bf T}_1$ as an empty set. Therefore,
\begin{equation}
\label{eq:periodrelation}
\forall 1 \leq i \leq k - 1, \quad t_{i+1}^* - t_i^* = \beta t_i^* U_i,
\end{equation}
which implies that
\begin{equation}
\label{eq:2ndperiodrelation}
\dfrac{t_{i+1}^*}{t_i^*} = \beta U_i + 1.
\end{equation}
  Moreover, 
\begin{equation}
\label{eq:3rdperiodrelation}
 \dfrac{t_i^*}{t_k^*}= \prod_{j=i}^{k-1}\frac{t_j^*}{t_{j+1}^*} = \frac{1}{\prod_{j=i}^{k-1} (\beta U_j + 1)}.
\end{equation}
The above extreme point solution is always feasible in the linear programming of Eq.~\eqref{eq:lp-framework-constrained} since  
\begin{equation}
  0 < \dfrac{t_i^*}{t_k^*},  \;\;\;\forall i=1,2,\ldots,k-1.
\end{equation}
By Eq.~\eqref{eq:periodrelation}, we have
  \begin{equation}
\label{eq:alpha+beta}
    (\alpha
    +\beta)U_i t_i^* = (t_{i+1}^*-t_i^*) (\frac{\alpha}{\beta}+1).
  \end{equation}
  Therefore,  in this extreme point solution, the objective function of Eq.~\eqref{eq:lp-framework-constrained} is 
  \begin{align}
    &t_k^* - \sum_{i=1}^{k-1} (\alpha
    +\beta)U_it_i^*= t_k^* - (t_k^*-t_1^*) (\frac{\alpha}{\beta}+1)\nonumber\\
 = & t_k^*\left(\frac{t_1^*}{t_k^*}(\frac{\alpha}{\beta}+1)-\frac{\alpha}{\beta}\right) \overset{1}{=}
 t_k^*\left(\frac{\frac{\alpha}{\beta}+1}{\prod_{j=1}^{k-1} (\beta U_j + 1)}-\frac{\alpha}{\beta}\right),
  \end{align}
where the last equality $\overset{1}{=}$ is due to Eq.~\eqref{eq:3rdperiodrelation} when $i$ is $1$.

{\bf [Step 3:]} From Step 1, a feasible extreme point solution is with ${\bf T}_1\cup{\bf T}_2=\setof{\tau_1, \tau_2, \ldots, \tau_{k-1}}$ and ${\bf T}_1\cap {\bf T}_2 = \emptyset$.
 As a result, we can simply drop all the tasks in ${\bf T}_1$ and use the remaining tasks in ${\bf T}_2$ by adopting the same procedures from Eq.~\eqref{eq:periodrelation} to Eq.~\eqref{eq:alpha+beta} in Step 2. The resulting objective function of this extreme point solution for the linear programming in
    Eq.~\eqref{eq:lp-framework-constrained} is $t_k^*(\frac{\frac{\alpha}{\beta}+1}{\prod_{\tau_j \in {\bf T}_2} (\beta U_j + 1)} - \frac{\alpha}{\beta}) \geq t_k^*(\frac{\frac{\alpha}{\beta}+1}{\prod_{j=1}^{k-1} (\beta U_j + 1)}-\frac{\alpha}{\beta})$ due to the fact that $\prod_{\tau_j \in {\bf T}_2} (\beta U_j + 1) \leq \prod_{j=1}^{k-1} (\beta U_j + 1)$.

    Therefore, for a given $s \geq 0$, i.e., $t_k^* \geq t_k$, all the other extreme points of
    Eq.~\eqref{eq:lp-framework-constrained} were dominated by the one
    specified in Step 2.  
  By the above analysis, if   $\frac{\frac{\alpha}{\beta}+1}{\prod_{j=1}^{k-1} (\beta U_j +
      1)}-\frac{\alpha}{\beta} > 0$, we know that
    $(t_k+s)(\frac{\frac{\alpha}{\beta}+1}{\prod_{j=1}^{k-1} (\beta
      U_j + 1)}-\frac{\alpha}{\beta})$ is an increasing function of
    $s$, in which the minimum happens when $s$ is $0$. As a result, we
    reach the conclusion of the lemma.
\end{proof}

We now provide two extended lemmas based on
Lemma~\ref{lemma:framework-constrained}, used for given $\alpha$ and
$\beta$ when $\alpha_i \leq \alpha$ and $\beta_i \leq \beta$ for
$i=1,2,\ldots,k-1$. Their proofs are in Appendix B.

\begin{lemma}
\label{lemma:framework-totalU-constrained}
For a given $k$-point effective schedulability test of a scheduling
algorithm, defined in
Definition~\ref{def:kpoints},
in which $0 < t_k$ and $0 < \alpha_i \leq \alpha$ and $0 < \beta_i \leq \beta$ for any
$i=1,2,\ldots,k-1$, task $\tau_k$ is schedulable by the scheduling
algorithm if 
\begin{equation}
\label{eq:schedulability-totalU-constrained}
\frac{C_k}{t_k} + \sum_{i=1}^{k-1}U_i \leq \begin{cases}
  1, & (\alpha+\beta)^{\frac{1}{k}} < 1\\
  (k-1)\frac{\left((1+\frac{\beta}{\alpha})^{\frac{1}{k-1}}-1\right)}{\beta}, & (\alpha+\beta)^{\frac{1}{k}} < \alpha\\
\frac{(k-1)((\alpha+\beta)^{\frac{1}{k}}-1)+((\alpha+\beta)^{\frac{1}{k}}-\alpha)}{\beta}  & \mbox{ otherwise}.
\end{cases}
\end{equation}
\end{lemma}

\begin{lemma}
\label{lemma:framework-totalU-exclusive}
For a given $k$-point effective schedulability test of a scheduling
algorithm, defined in
Definition~\ref{def:kpoints},
in which $0 < t_k$ and $0 < \alpha_i \leq \alpha$ and $0 < \beta_i \leq \beta$ for any
$i=1,2,\ldots,k-1$,
task $\tau_k$ is schedulable by the scheduling
algorithm if 
\begin{equation}
\label{eq:schedulability-totalU-exclusive}
\beta \sum_{i=1}^{k-1}U_i \leq \ln(\frac{\frac{\alpha}{\beta}+1}{\frac{C_k}{t_k}+\frac{\alpha}{\beta}}).
\end{equation}
\end{lemma}



We also construct the following Lemma~\ref{lemma:framework-general},  c.f. Appendix B for the proof, as the most powerful method for a concrete task system. 
Throughout the paper, we will not build theorems and corollaries based on 
Lemma~\ref{lemma:framework-general}\ifbool{techreport}{ but we will
  evaluate how it performs in the experimental section.}{, the
  performance evaluation can be found in the report \cite{DBLP:journals/corr/abs-1501.07084}.}

\begin{lemma}
\label{lemma:framework-general}
For a given $k$-point effective schedulability test of a fixed-priority scheduling
algorithm, defined in
Definition~\ref{def:kpoints}, task $\tau_k$ is schedulable by the scheduling algorithm,
in which $0 < t_k$ and $0 < \alpha_i$ and $0 < \beta_i$ for any
$i=1,2,\ldots,k-1$, 
 if
the following condition holds
\begin{equation}
\label{eq:schedulability-general}
0 < \frac{C_k}{t_k} \leq 1 -  \sum_{i=1}^{k-1}  \frac{U_i(\alpha_i
  +\beta_i)}{\prod_{j=i}^{k-1} (\beta_jU_j + 1)}.
\end{equation}
\end{lemma}

\noindent{\bf{\large Remarks and how to use the framework}:}
After presenting the \framework{} framework, here, we explain how to use the \framework{} framework and summarize how we
plan to demonstrate its applicability in several task models in the following sections.
The \framework{} framework relies on the users to index the tasks properly and define $\alpha_i$ and $\beta_i$ as constants for $i=1,2,\ldots,k-1$ based on Eq.~\eqref{eq:precodition-schedulability}. The set $\setof{t_1, t_2, \ldots, t_{k-1}}$ in Definition~\ref{def:kpoints} is used only
for the users to define those constants, where $t_k$ is usually defined to be the interval length of the original schedulability test, e.g., $D_k$ in Eq.~\eqref{eq:exact-test-constrained-deadline}.
Therefore, the \framework{} framework can only be applicable when $\alpha_i$ and $\beta_i$ are well-defined. These constants
depend on the task models and the task parameters.

The choice of good parameters $\alpha_i$ and $\beta_i$ affects the
quality of the resulting schedulability bounds in Lemmas~\ref{lemma:framework-constrained} to \ref{lemma:framework-totalU-exclusive}.
However, deriving the \emph{good}  settings of $\alpha_i$ and $\beta_i$ is actually not the focus of this paper. The framework does not care how the parameters $\alpha_i$ and $\beta_i$ are obtained. The framework simply derives the bounds according to the given parameters $\alpha_i$ and $\beta_i$, regardless of the settings of $\alpha_i$ and $\beta_i$. The correctness of the settings of $\alpha_i$ and $\beta_i$ is not verified by the framework. 

The ignorance of the settings of $\alpha_i$ and $\beta_i$ actually leads to the elegance and the generality of the framework, which works as long as Eq.~\eqref{eq:precodition-schedulability} can be successfully constructed for the sufficient schedulability test of task $\tau_k$ in a fixed-priority scheduling policy.
With the availability of the \framework{} framework, the hyperbolic bounds or utilization bounds can be automatically derived by adopting  Lemmas \ref{lemma:framework-constrained} to \ref{lemma:framework-totalU-exclusive} as long as the safe upper bounds $\alpha$ and $\beta$  to cover all the possible settings of $\alpha_i$ and $\beta_i$ for the schedulability test in Eq.~\eqref{eq:precodition-schedulability} can be derived, regardless of the task model or the platforms.

The other approaches in \cite{journals/tc/LeeSP04,DBLP:dblp_journals/tc/BurchardLOS95,HanTyan-RTSS97} also have similar observations by testing only several time points in the TDA schedulability analysis based on Eq.~\eqref{eq:exact-test-constrained-deadline} in their problem formulations. Specifically, the problem formulations in \cite{DBLP:dblp_journals/tc/BurchardLOS95,HanTyan-RTSS97} are based on non-linear programming by selecting several good points under certain constraints. Moreover, the linear-programming problem formulation in \cite{journals/tc/LeeSP04} considers $U_i$ as variables and $t_i$ as constants and solves the corresponding linear programming analytically.
However, as these approaches in 
\cite{journals/tc/LeeSP04,DBLP:dblp_journals/tc/BurchardLOS95,HanTyan-RTSS97} seek for the total utilization bounds, they have limited applications and are less flexible. For example, they are typically not applicable directly when considering sporadic real-time tasks with arbitrary deadlines or multiprocessor systems. Here, we are more flexible in the  \framework{} framework. 
For task $\tau_i$, after $\alpha_i$ and $\beta_i$ or their safe upper bounds $\alpha$ and $\beta$ are derived, we completely drop out the dependency to the periods and models inside \framework{}.

\ifbool{techreport}{
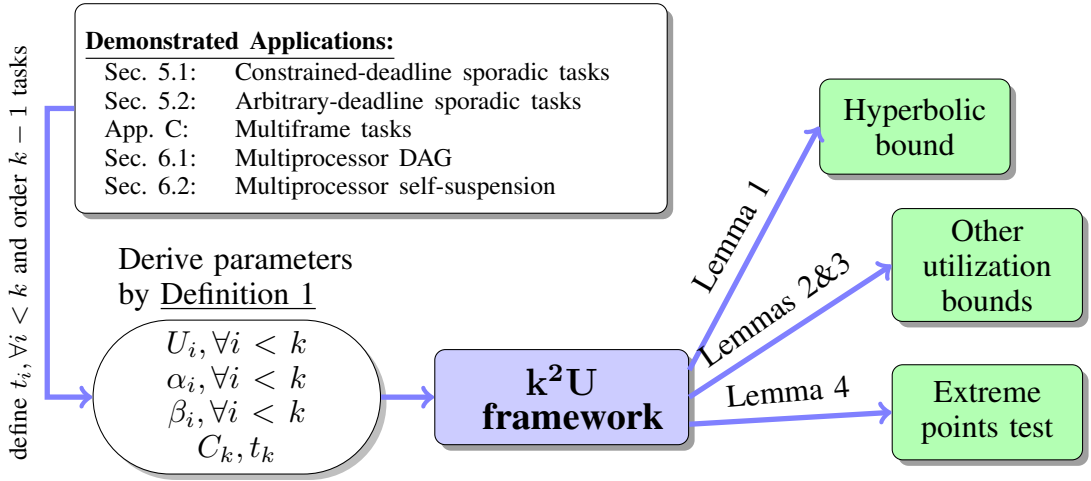
\begin{figure*}[t]
} {
\begin{figure}[t]
}
	\begin{center}
\ifbool{techreport}{
      \begin{tikzpicture}[scale=1.2,transform shape]
} {
      \begin{tikzpicture}[scale=0.75,transform shape]
}
        \path \practica {1}{\underline{\bf Demonstrated Applications:}
          \begin{tabular}{ll}
          Sec. 5.1:& Constrained-deadline sporadic tasks\\
          Sec. 5.2:& Arbitrary-deadline sporadic tasks\\
          App. C\ifbool{techreport}{}{ \cite{DBLP:journals/corr/abs-1501.07084}}:& Multiframe tasks\\
          Sec. 6.1:& Multiprocessor DAG\\
          Sec. 6.2:& Multiprocessor self-suspension\\            
          \end{tabular}
        };
        \path (p1.west)+(1.8,-3.2) node(p2)[materia, rounded rectangle,text width=6em]{{\bf $U_i, \forall i < k$\\$\alpha_i, \forall i<k$\\$\beta_i, \forall i <k$\\$C_k, t_k$}}; 
        \path (p2.north)+(1.5,0.45) node[text width=16em]{Derive parameters\\by \underline{Definition 1}};
        \path (p2.east)+(2,0) node(p3)[practica,fill=blue!20,text width=5em,text centered]{\large{\bf \framework{}\\ framework}};

        \path (p3.east)+(2.5,+3) node(p4)[practica,fill=green!30,minimum width=6em,text width=5em,text centered]{Hyperbolic bound}; 
        \path (p4.south)+(0.8,-1) node(p5)[practica,fill=green!30,minimum width=6em,text width=5em,text centered]{Other utilization bounds}; 
        \path (p5.south)+(0,-1) node(p6)[practica,fill=green!30,minimum width=6em,text width=5em,text centered]{Extreme points test}; 
        
        \path [line, ->] (p1.west) -- +(-0.3,0) node[black, rotate=90, yshift=0.3cm, xshift=-1.6cm]{\footnotesize{define $t_i, \forall i < k$ and order $k-1$ tasks}} -- + (-0.3, -3.2) -- (p2.west);
        \path [line, ->] (p2.east) -- (p3.west);
        \path [line, ->] (p3.east)+(0,0.3) -- node[rotate=60,yshift=0.3cm,black]{Lemma 1} (p4.west);
        \path [line, ->] (p3.east)+(0,0) -- node[rotate=35,yshift=0.3cm,black]{Lemmas 2\&3} (p5.west);
        \path [line, ->] (p3.east)+(0,-0.3) -- node[rotate=5,yshift=0.3cm,black]{Lemma 4} (p6.west);
      \end{tikzpicture}    
	\end{center}
\vspace{-2mm}
\caption{The \framework{} framework. }
\label{fig:framework}
\ifbool{techreport}{
\end{figure*}
} {
\end{figure}
}

We will demonstrate the applicability and generality of \framework{} by using the most-adopted sporadic real-time task model in Sec.~\ref{sec:application-FP}, multi-frame tasks in Appendix C\citetechreport{}
and multiprocessor scheduling in Sec.~\ref{sec:multiprocessor}, as illustrated in Figure~\ref{fig:framework}.
In all these cases, we can find reasonable settings of $\alpha$ and $\beta$ to provide better results or new results for schedulability tests, with respect to the utilization bounds, speed-up factors, or capacity augmentation factors, compared to the literature.  More specifically, (after certain reorganizations), we will greedily set $t_i$ as $\floor{\frac{t_k}{T_i}}T_i$ in all of the studied cases.\footnote{Setting $t_i$ as $\floor{\frac{t_k}{T_i}}T_i$ is actually the same in \cite{journals/tc/LeeSP04} for the sporadic real-time task model with implicit deadlines and the multi-frame task model when $T_i$ is given and $U_i$ is considered as a variable.}
Table \ref{tab:alpha-beta} summarizes the $\alpha_i$ and $\beta_i$ parameters derived in this paper, as well as an earlier result by Liu and Chen in \cite{RTSS14a} for self-suspending task models and deferrable servers.

\begin{table*}[t]
\renewcommand{\arraystretch}{1.5}
  \centering
  \scalebox{0.8}{
  \begin{tabular}{|p{7cm}|c|c|p{4cm}|}
    \hline
    \hline
    Model & $\alpha_i$ & $\beta_i$ & c.f.\\ 
    \hline
    \hline
  Uniprocessor Sporadic Tasks& $\alpha_i=1$ & $0 < \beta_i \leq 1$ & Theorem~\ref{theorem:sporadic-general} and Corollary~\ref{col:arbitrary-general}\\
    \hline
   Multiprocessor Global RM for Different Models & $0 < \alpha_i \leq
   \frac{2}{M}$ & $0 < \beta_i\leq \frac{1}{M}$ &
   Theorems~\ref{thm:multiprocessor-DAG} -
   \ref{thm:multiprocessor-suspension}\\
    \hline
  Uniprocessor  Multi-frame Tasks& $0 < \alpha_i \leq 1$ & $0 < \beta_i\leq \frac{\phi_i(2)-\phi_i(1)}{\phi_i(1)}$ & \ifbool{techreport}{Theorem~\ref{thm:multiframe}}{Theorem 7 in \cite{DBLP:journals/corr/abs-1501.07084}}\\

    \hline
   Uniprocessor Self-Suspending Tasks& $0 < \alpha_i\leq 2$ & $0 < \beta_i\leq 1$ & Theorems 5 and 6 in \cite{RTSS14a}, implicitly\\
    \hline
  \end{tabular}
}
  \caption{\small The applicability of \framework{} for different task models: their $\alpha_i$ and $\beta_i$ parameters, where $M$ is the number of processors in global RM scheduling and  $\phi_i(\ell)$ is the maximum of the sum of the execution time of
    any $\ell$ consecutive frames of task $\tau_i$ in the multi-frame model.}
  \label{tab:alpha-beta}
\end{table*}

\section{Applications for Fixed-Priority Scheduling}
\label{sec:application-FP}

This section provides demonstrations on how to use the \framework{}
framework to derive efficient schedulability analysis for sporadic
task systems in uniprocessor systems. We will consider
constrained-deadline systems first in
Sec.~\ref{sec:application-constrained}
and explain how to extend to
arbitrary-deadline systems in Sec.~\ref{sec:application-arbitrary}. For the rest of this section, we will implicitly assume $C_k > 0$.

\subsection{Constrained-Deadline Systems}
\label{sec:application-constrained}

For a specified fixed-priority scheduling algorithm, let $hp(\tau_k)$ be the set
of tasks with higher priority than $\tau_k$. We now classify the task
set $hp(\tau_k)$ into two subsets:
\begin{itemize}
\item $hp_1(\tau_k)$ consists of the higher-priority tasks with periods
  smaller than $D_k$.
\item $hp_2(\tau_k)$ consists of the higher-priority tasks with periods
  larger than or equal to $D_k$.
\end{itemize}

For any $0 < t \leq D_k$,
we know that a safe upper bound on the interference due to higher-priority tasks is given by
\[
\sum_{\tau_i \in hp(\tau_k)} \ceiling{\frac{t}{T_i}}C_i = \sum_{\tau_i \in hp_2(\tau_k)} C_i + \sum_{\tau_i \in hp_1(\tau_k)} \ceiling{\frac{t}{T_i}}C_i .
\]
As a result, the schedulability test in
Eq.~\eqref{eq:exact-test-constrained-deadline}  is equivalent to the
verification of the existence of $0 < t \leq D_k$ such that
\begin{equation}
  \label{eq:exact-test-constrained-deadline-2}
 C_k + \sum_{\tau_i \in hp_2(\tau_k)} C_i + \sum_{\tau_i \in hp_1(\tau_k)} \ceiling{\frac{t}{T_i}}C_i \leq t.  
\end{equation}
We can then create a virtual sporadic task $\tau_k'$ with execution
time $C_k'=C_k + \sum_{\tau_i \in hp_2(\tau_k)} C_i$, relative
deadline $D_k'=D_k$, and period $T_k'=D_k$. It is
clear that the schedulability test to verify the schedulability of
task $\tau_k'$ under the interference of the higher-priority tasks
$hp_1(\tau_k)$ is the same as that of task $\tau_k$ under the
interference of the higher-priority tasks $hp(\tau_k)$.

Therefore, with the above analysis, we can use the \framework{}
framework in Sec.~\ref{sec:framework} as in the following theorem.

\begin{theorem}
\label{theorem:sporadic-general}
Task $\tau_k$ in a sporadic task system with constrained deadlines is
schedulable by the fixed-priority scheduling algorithm if
\begin{equation}
\label{eq:schedulability-sporadic-any-a}
(\frac{C_k'}{D_k}+1) \prod_{\tau_j \in hp_1(\tau_k)} (U_j + 1)\leq 2
\end{equation}
or
\begin{equation}
\label{eq:schedulability-sporadic-any-b}
\frac{C_k'}{D_k}+  \sum_{\tau_j \in hp_1(\tau_k)}U_j\leq  k(2^{\frac{1}{k}}-1).
\end{equation}
\end{theorem}
\begin{proof}
For notational brevity, suppose that there are $k-1$ tasks in
$hp_1(\tau_k)$.\footnote{When $hp_2(\tau_k)$ is not empty, there are $k-1-|hp_2(\tau_k)|$ tasks in $hp_1(\tau_k)$.
To be notationally precise, we can denote the number of tasks in $hp_1(\tau_k)$ by a new symbol $k^*$. Since  $k^* \leq k$, we know $k^*(2^{\frac{1}{k^*}}-1) \geq  k(2^{\frac{1}{k}}-1)$ as a safe bound in Eq.~\eqref{eq:schedulability-sporadic-any-b}. However, this may make the notations too complicated. We have decided to keep it as $k-1$ for the sake of notational brevity. }  
Now, we index the higher-priority tasks in
$hp_1(\tau_k)$ to form the corresponding $\tau_1, \tau_2, \ldots,
\tau_{k-1}$. The $k-1$ higher-priority tasks in $hp_1(\tau_k)$ are
ordered to ensure that the arrival times, i.e., $\floor{\frac{D_k}{T_i}}T_i$,  of the last jobs no later than $D_k$
are in a non-decreasing order.  That is, with the above indexing of
the higher-priority tasks in $hp_1(\tau_k)$, we have
$\floor{\frac{D_k}{T_i}}T_i \leq
\floor{\frac{D_k}{T_{i+1}}}T_{i+1}$ for $i=1,2,\ldots,k-2$.
Now, we set $t_i$ as $\floor{\frac{D_k}{T_i}}T_i$ for
$i=1,2,\ldots,k-1$, and $t_k$ as $D_k$. Due to the fact that $T_i \leq D_k$
for
$i=1,2,\ldots,k-1$, we know that $t_i > 0$.

Therefore, 
for a given $t_j$ with $j=1,2,\ldots,k$, the demand requested up to time $t_j$ is at most  
\begin{align}
  & C_k + \sum_{\tau_i \in hp_2(\tau_k)} C_i + \sum_{\tau_i \in hp_1(\tau_k)} \ceiling{\frac{t_j}{T_i}}C_i\nonumber\\
  = & C_k' +  \sum_{i=1}^{k-1} \ceiling{\frac{t_j}{T_i}}C_i
  \leq   C_k' + \sum_{i=1}^{k-1} \frac{t_i}{T_i}C_i +  \sum_{i=1}^{j-1} C_i,\nonumber
\end{align}
where the inequality comes from the indexing policy defined above,
i.e., $\ceiling{\frac{t_j}{T_i}} \leq \frac{t_i}{T_i}+1$ if $j
> i$  and $\ceiling{\frac{t_j}{T_i}} \leq
\frac{t_i}{T_i}$ if $j \leq i$.
Since $T_i < D_k$ for any task $\tau_i$ in $hp_1(\tau_k)$, we know
  that $t_i \geq T_i$. 
Instead of testing all the $t$ values in
Eq.~(\ref{eq:exact-test-constrained-deadline-2}), we only apply the
test for these $k$ different $t_i$ values, which is equivalent to the
test of the existence of $t_j$ such that
\begin{equation}
\label{eq:k-point-sporadic}
  C_k'
  + \sum_{i=1}^{k-1} t_i U_i + \sum_{i=1}^{j-1}  \frac{T_i}{t_i}  t_i U_i \leq t_j.
\end{equation}
Therefore, we reach the schedulability in the form of Eq.~\eqref{eq:precodition-schedulability} under the setting of $\alpha_i$ to $1$ and $\beta_i$ to $\frac{T_i}{t_i} \leq 1$ (due to $t_i \geq T_i$), for $i=1,2,\ldots,k-1$.  By taking $\beta_i \leq 1$ and $\alpha_i = 1$ for
  $i=1,2,\ldots,k-1$ in Lemmas~\ref{lemma:framework-constrained}
  and~\ref{lemma:framework-totalU-constrained}, we reach the
  conclusion.
\end{proof}


When RM priority ordering is applied for an implicit-deadline
task system, $C_k'$ is equal to $C_k$ and $\frac{C_k'}{D_k}$ is equal
to $U_k$. For such a case, the condition in
Eq.~\eqref{eq:schedulability-sporadic-any-a} is the same as the
hyperbolic bound provided in \cite{bini2003rate}, and the condition in
Eq.~\eqref{eq:schedulability-sporadic-any-b} is the same as least
utilization upper bound in \cite{liu1973scheduling}.

The schedulability test in
Theorem~\ref{theorem:sporadic-general} may seem pessimistic at first glance. We evaluate the quality of the
schedulability test in Theorem~\ref{theorem:sporadic-general} by
quantifying the speed-up factor with respect to the optimal schedule
(i.e., EDF scheduling in such a case). We show that the speed-up
factor of the schedulability test in
Eq.~\eqref{eq:schedulability-sporadic-any-a} is $1.76322$, which is
the same as the lower bound and upper bound of DM as shown in
\cite{DBLP:journals/rts/DavisRBB09}. The speed-up factor of DM,
regardless of the schedulability tests, obtained by Davis et
al. \cite{DBLP:journals/rts/DavisRBB09} is the same as our result. Our
result is thus stronger, as we show that the factor already holds when
using the schedulability test in Eq.~\eqref{eq:schedulability-sporadic-any-a} in the following theorem.

\begin{theorem}
  \label{thm:speedup-DM}
  The speed-up factor of the schedulability test in
  Eq.~\eqref{eq:schedulability-sporadic-any-a} under DM scheduling for
  constrained-deadline tasks is $1.76322$ with respect to EDF.
\end{theorem}

The proof of Theorem~\ref{thm:speedup-DM} in the Appendix B, which is much simpler than the proof in \cite{DBLP:journals/rts/DavisRBB09}, can also be considered as an alternative proof of the speed-up factor of DM.

\begin{corollary}
  \label{col:implicit-general-ratio}
  Suppose that $f \cdot T_i \leq D_k$ for any higher priority task
  $\tau_i$ in $hp_1(\tau_k)$, where $f$ is a positive integer.  Task $\tau_k$
  in a constrained-deadline sporadic task system is schedulable by the fixed-priority
  scheduling algorithm if
\begin{equation}
\label{eq:schedulability-sporadic-implicit-ratio-a}
(\frac{C_k'}{f \cdot D_k}+1) \prod_{\tau_j \in hp_1(\tau_k)} (\frac{U_j}{f} + 1)\leq \frac{f+1}{f}
\end{equation}
or
\begin{equation}
  \label{eq:schedulability-sporadic-implicit-ratio-b}
\frac{C_k'}{D_k}+  \sum_{\tau_j \in hp_1(\tau_k)}U_j\leq f k\left(\left({\frac{f+1}{f}}\right)^{\frac{1}{k}}-1\right).
\end{equation}
\end{corollary}
\begin{proof}
  This is based on the same proof as Theorem~\ref{theorem:sporadic-general}, by taking the fact that $\frac{t_i}{T_i} = \frac{\floor{\frac{D_k}{T_i}}T_i}{T_i} \geq f$. 
  In the $k$-point effective schedulability test, we can set
  $\alpha_i$ to $1$, $\beta_i = \frac{T_i}{t_i} \leq \frac{1}{f}$.
  Therefore, we have $\alpha_i \leq \alpha=1$,
  $\beta_i\leq\beta=\frac{1}{f}$, and $\frac{\alpha}{\beta}$ is $f$.
  By adopting
  Lemma~\ref{lemma:framework-constrained}, we know that task $\tau_k$
  is schedulable under the scheduling policy if
\ifbool{techreport}{
\[}{$}
 (\frac{C_k'}{D_k} + f) \prod_{{\tau_j \in hp_1(\tau_k)}} (\frac{U_j}{f} + 1)\leq f+1,
\ifbool{techreport}{
\]}{$}
  which is the same as the condition in
  Eq.~\eqref{eq:schedulability-sporadic-implicit-ratio-a} by dividing
  both sides by $f$. By using a similar argument and applying
  Lemma~\ref{lemma:framework-totalU-constrained}, we can reach the
  condition in Eq.~\eqref{eq:schedulability-sporadic-implicit-ratio-b}. 
\end{proof}
\ifbool{techreport}{
Note that the right-hand side of
Eq.~\eqref{eq:schedulability-sporadic-implicit-ratio-b} converges to $f
\ln(\frac{f+1}{f})$ when $k$ goes to $\infty$.}{}




\subsection{Arbitrary-Deadline Systems}
\label{sec:application-arbitrary}

We now further explore how to use the proposed framework to perform
the schedulability analysis for arbitrary-deadline task sets.
The
exact schedulability analysis for arbitrary-deadline task sets under
fixed-priority scheduling has been developed in
\cite{DBLP:conf/rtss/Lehoczky90}. The schedulability analysis is to
use a \emph{busy-window} concept to evaluate the worst-case response
time. That is, we release all the higher-priority tasks together with
task $\tau_k$ at time $0$ and all the subsequent jobs are released as
early as possible by respecting to the minimum inter-arrival time. The
busy window finishes when a job of task $\tau_k$ finishes before the
next release of a job of task $\tau_k$. It has been shown in
\cite{DBLP:conf/rtss/Lehoczky90} that the worst-case response time of
task $\tau_k$ can be found in one of the jobs of task $\tau_k$ in the
busy window.

Therefore, a simpler sufficient schedulability test for a task
$\tau_k$ is to test whether the length of the busy window is within
$D_k$.  If so, all invocations of task $\tau_k$ released in the busy
window can finish before their relative deadline. Such an observation
has also been adopted in \cite{conf:/rtns09/Davis}. Therefore, a
sufficient test is to verify whether\footnote{This analysis is pretty pessimistic. 
But, as our objective in this paper is to show the applicability and generality of \framework{}, 
how to use \framework{} to get tighter analysis for arbitrary-deadline task systems is not our focus in this paper. Some evaluations can be found in \cite{DBLP:journals/corr/framework-compare}.  }
\begin{equation}
  \label{eq:sufficient-test-arbitrary-deadline}
\exists t \mbox{ with } 0 < t \leq D_k {\;\; and \;\;} \ceiling{\frac{t}{T_k}}C_k +
\sum_{\tau_i \in hp(\tau_k)} \ceiling{\frac{t}{T_i}}C_i \leq t. 
\end{equation}
If the condition in Eq.~\eqref{eq:sufficient-test-arbitrary-deadline} holds,
it implies that the busy window (when considering task $\tau_k$) is no
more than $D_k$, and, hence, task $\tau_k$ has worst-case response
time no more than $D_k$.

If $D_k \leq T_k$, the analysis in
Sec.~\ref{sec:application-constrained} can be directly applied. If
$D_k > T_k$, we need to consider the length of the busy-window for
task $\tau_k$ as shown above. For the rest of this section, we will
focus on the case $D_k > T_k$.
We can rewrite
Eq.~\eqref{eq:sufficient-test-arbitrary-deadline} to use a more
pessimistic condition by releasing the workload
$\ceiling{\frac{D_k}{T_k}}C_k$ at time $0$. That is, if
\begin{equation}
  \label{eq:sufficient-test-arbitrary-deadline-pessimistic}
\exists t \mbox{ with } 0 < t \leq D_k {\;\; and \;\;} \ceiling{\frac{D_k}{T_k}}C_k +
\sum_{\tau_i \in hp(\tau_k)} \ceiling{\frac{t}{T_i}}C_i \leq t, 
\end{equation}
then, the length of the busy window for task $\tau_k$ is no more than
$D_k$. Again, similar to the strategy
we use in Sec.~\ref{sec:application-constrained}, we classify the
tasks in $\tau_i \in hp(\tau_k)$ into two sets $hp_1(\tau_k)$ and
$hp_2(\tau_k)$ with the same definition.

Similarly, we can then create a virtual sporadic task $\tau_k'$ with
execution time $C_k' = \ceiling{\frac{D_k}{T_k}}C_k + \sum_{\tau_i \in
  hp_2(\tau_k)} C_i$, relative deadline $D_k'=D_k$, and period $T_k'=D_k$.
For notational brevity, suppose that there are $k-1$ tasks in
$hp_1(\tau_k)$.  Now, we index the higher-priority tasks in
$hp_1(\tau_k)$ to form the corresponding $\tau_1, \tau_2, \ldots,
\tau_{k-1}$.  In the above definition of the busy window concept,
$\floor{\frac{D_k}{T_i}}T_i$ is the arrival time of the last job
of task $\tau_i$ released no later than $D_k$. The $k-1$ higher-priority
tasks in $hp_1(\tau_k)$ are ordered to ensure that the
arrival times of the last jobs before $D_k$ are in a non-decreasing
order.  Moreover, $t_k$ is the specified testing point $D_k$.  Instead
of testing all the $t$ values in
Eq.~\eqref{eq:sufficient-test-arbitrary-deadline-pessimistic}, we only
apply the test for these $k$ different $t_i$ values, which is
equivalent to the test of the existence of $t_j$ such that
Eq.~\eqref{eq:k-point-sporadic} holds, where $\alpha_i \leq 1$ and
$\beta_i = \frac{T_i}{t_i} \leq 1$ for $i=1,2,\ldots,k-1$, similar to the proof of Theorem~\ref{theorem:sporadic-general}.
\ifbool{techreport}{

Therefore, we can then use the \framework{} framework, i.e., Lemmas
\ref{lemma:framework-constrained} to
\ref{lemma:framework-totalU-exclusive} to test the schedulability of
task $\tau_k$. }{}
The following corollary comes from a similar argument as in
Sec.~\ref{sec:application-constrained}.
\begin{corollary}
\label{col:arbitrary-general}
Task $\tau_k$ in a sporadic arbitrary-deadline task system is schedulable by the
fixed-priority scheduling algorithm if
Eq.~\eqref{eq:schedulability-sporadic-any-a} or
\eqref{eq:schedulability-sporadic-any-b} holds,
in which there are $k-1$ higher priority tasks in $hp_1(\tau_k)$ and $C_k'$ is defined as $\ceiling{\frac{D_k}{T_k}}C_k + \sum_{\tau_i
  \in hp_2(\tau_k)} C_i$.
\end{corollary}



\begin{corollary}
  \label{col:utilization-RM-arbitrary}
  Suppose that $f \cdot T_i \leq D_k$ for any higher task $\tau_i$ in the
  task system and $f\cdot T_k \leq D_k$, where $f$ is a positive integer.  Task $\tau_k$
  in a sporadic task system is schedulable by using RM, i.e., $T_i \leq T_{i+1}$, if
\begin{equation}
\label{eq:schedulability-arbitrary-rm-ratio-a}
(\frac{U_k}{f}+1) \prod_{j=1}^{k-1} (\frac{U_j}{f} + 1)\leq \frac{f+1}{f}
\end{equation}
or
\begin{equation}
  \label{eq:schedulability-arbitrary-rm-ratio-b}
\sum_{j=1}^{k}U_j \leq  f k\left(\left({\frac{f+1}{f}}\right)^{\frac{1}{k}}-1\right).
\end{equation}
\end{corollary}
\begin{proof}
  Based on the above assumption $f \cdot T_i \leq D_k$ for any higher task $\tau_i$ in the
  task system, $f\cdot T_k \leq D_k$, and the rate
  monotonic scheduling, we can safely set $\alpha_i$ to $1$, $\beta_i$
  to $\frac{1}{f}$, and $C_k'$ to $f \cdot C_k$. Note that $\frac{C_k'}{f\cdot t_k} \leq \frac{U_k}{f}$ in this case. Therefore, by
  adopting Lemma~\ref{lemma:framework-constrained} and
  Lemma~\ref{lemma:framework-totalU-constrained}, we reach the
  conclusion, as in the proof of Corollary~\ref{col:implicit-general-ratio}.
\end{proof}



\section{Application for Multiprocessor Scheduling}
\label{sec:multiprocessor}

It may seem, at first glance, that the \framework{} framework only
works for uniprocessor systems. We demonstrate in this section how to
use the framework in multiprocessor global RM scheduling when considering
implicit-deadline, DAG, and self-suspending task systems. The methodology can also be extended to
handle constrained-deadline systems.

In multiprocessor global scheduling, we consider that the system has
$M$ identical processors, each with the same
computation power. Moreover, there is a global queue and a global
scheduler to dispatch the jobs. We consider only global RM scheduling,
in which the priority of the tasks is defined based on RM. At any
time, the $M$-highest-priority jobs in the ready queue are dispatched
and executed on these $M$ processors. 

Global RM in general does not have good utilization bounds.
 However, if we constrain the total utilization
$\sum_{\tau_i} \frac{C_i}{M T_i} \leq \frac{1}{b}$
and the maximum utilization $\max_{\tau_i} \frac{C_i}{T_i} \leq
\frac{1}{b}$, it is possible to provide the schedulability guarantee
of global RM by setting $b$ to $3-\frac{1}{M}$ \cite{DBLP:conf/rtss/AnderssonBJ01,DBLP:conf/rtss/Baker03,DBLP:conf/opodis/BertognaCL05}. Such a factor $b$ has
been recently named as a \emph{capacity augmentation factor} \cite{Li:ECRTS14}.

We will use the following
time-demand function $W_i(t)$ for the simple sufficient schedulability analysis:
\begin{equation}
  \label{eq:W_i-multiprocessor}
W_i(t) = (\ceiling{\frac{t}{T_i}}-1)C_i + 2C_i,
\end{equation}
which can be imagined as if the \emph{carry-in} job is fully carried into the analysis interval \cite{DBLP:conf/rtss/GuanSYY09}.
That is, we allow the first release of task $\tau_i$ to be inflated by a
factor $2$, whereas the other jobs of task $\tau_i$ have the same
execution time $C_i$.
Again, we consider testing the schedulability of task $\tau_k$ under
global RM, in which there are $k-1$ higher-priority tasks $\tau_1,
\tau_2, \ldots, \tau_{k-1}$.  We have the following schedulability
condition for global RM.
\begin{lemma}
\label{lemma:gRM-sufficient}
Task $\tau_k$ is schedulable under
  global RM on $M$ identical processors, if
  \begin{equation}
    \label{eq:gRM-sufficient}
\exists t \mbox{ with } 0 < t \leq T_k {\;\; and \;\;} C_k +\sum_{i=1}^{k-1} 
\frac{W_i(t)}{M} \leq t, 
  \end{equation}
  where $W_i(t)$ is defined in Eq.~\eqref{eq:W_i-multiprocessor}.
\end{lemma}
\begin{proof}
  This has been shown in Sec. 3.2 (The Basic Multiprocessor Case)
  in \cite{DBLP:conf/rtss/GuanSYY09}.
\end{proof}

\begin{theorem}
\label{thm:multiprocessor-GRM}
Task $\tau_k$ in a sporadic implicit-deadline task system is
schedulable by global RM on $M$ processors if
\begin{equation}
\label{eq:schedulability-GRM}
 (\frac{C_k}{T_k}+2)\prod_{j=1}^{k-1} (\frac{U_j}{M} + 1)\leq 3,
\end{equation}
or
\begin{equation}
\label{eq:schedulability-augmentation-GRM}
\sum_{j=1}^{k-1} \frac{U_j}{M}\leq \ln{\frac{3}{\frac{C_k}{T_k}+2}}.
\end{equation}
\end{theorem}
\begin{proof}
  Let $t_i$ be $\floor{\frac{T_k}{T_i}}T_i$ for $i=1,2,\ldots,k$, and
  reindex the tasks such that $t_1 \leq t_2 \leq \ldots \leq t_k$.  By
  testing only these $k$ points in the schedulability test in
  \eqref{eq:gRM-sufficient} results in a $k$-point effective
  schedulability test with $\alpha_i \leq \frac{2}{M}$ and $\beta_i
  \leq \frac{1}{M}$. Therefore, we can adopt the \framework{}
  framework. By Lemma~\ref{lemma:framework-constrained} and
  Lemma~\ref{lemma:framework-totalU-exclusive}, we have concluded the
  proof.
\end{proof}

Note that Theorem~\ref{thm:multiprocessor-GRM}  is not superior to the known analysis for 
sporadic task systems \cite{DBLP:conf/rtss/AnderssonBJ01,DBLP:conf/rtss/Baker03,DBLP:conf/opodis/BertognaCL05},  as the schedulability condition in Lemma~\ref{lemma:gRM-sufficient} is too pessimistic. This is only used as the basis  to analyze a sporadic task set
with DAG tasks in Sec.~\ref{sec:global-DAG} and self-suspending tasks in Sec.~\ref{sec:global-SSS} when adopting global RM, in which we will demonstrate similar structures as used in Theorem~\ref{thm:multiprocessor-GRM}. We also demonstrate a tighter test in Appendix F\citetechreport{}
for improving the schedulability test of global RM for sporadic tasks.

\subsection{Global RM for DAG Task Systems}
\label{sec:global-DAG}

For multiprocessor scheduling, the DAG task model has been recently studied \cite{DBLP:conf/ecrts/BonifaciMSW13}. The utilization-based analysis can be found
 in \cite{Li:ECRTS14} and \cite{DBLP:conf/ecrts/BonifaciMSW13}.
Each task $\tau_i\in {\bf T}$ in a DAG task system is a parallel task.
Each task is characterized by its execution pattern, defined by a set
of directed acyclic graphs (DAGs). The execution time of a job of task
$\tau_i$ is one of the DAGs. Each node (subtask) in a DAG represents a
sequence of instructions (a thread) and each edge represents a
dependency between nodes.  A node (subtask) is \emph{ready} to be
executed when all its predecessors have been executed.  We will only
consider two parameters related to the execution pattern of task
$\tau_i$:
\begin{compactitem}
\item \emph{total execution time (or work)} $C_i$ of task $\tau_i$: This is
  the summation of the execution times of all the subtasks
  of task $\tau_i$ among all the DAGs of task $\tau_i$.  
\item \emph{critical-path length} $\Psi_i$ of task $\tau_i$: This is the length of
  the critical path among the given DAGs, in which each node is
  characterized by the execution time of the corresponding
  subtask of task $\tau_i$.  
\end{compactitem}
The analysis is based on the two given parameters $C_i$ and
$\Psi_i$. Therefore, we can also allow flexible DAG structures. That
is, jobs of a task may have different DAG structures, under the total
execution time constraint $C_i$ and the critical path length
constraint $\Psi_i$. Therefore, the model can also be applied for conditional sporadic DAG task systems \cite{DBLP:conf/ecrts/BaruahBM15}.
With the above definition, we have the following lemma, in which the proof is 
Appendix B\citetechreport{}.

\begin{lemma}
  \label{lemma:gRM-DAG-sufficient}
  Task $\tau_k$ in a sporadic DAG system with implicit deadlines is schedulable under global RM on $M$ identical
  processors, if
  \begin{equation}
    \label{eq:gRM-DAGsufficient}
\exists t \mbox{ with } 0 < t \leq T_k {\;\; and \;\;} \Psi_k+\frac{C_k-\Psi_k}{M}+
\sum_{i=1}^{k-1} 
\frac{W_i(t)}{M}\leq t,     
  \end{equation}
  where $W_i(t)$ is defined in Eq.~\eqref{eq:W_i-multiprocessor}.
 \end{lemma}

\begin{theorem}
\label{thm:multiprocessor-DAG}
Task $\tau_k$ in a sporadic DAG system with implicit deadlines is schedulable by global RM on
$M$ processors if
\begin{equation}
\label{eq:schedulability-DAG}
 (\frac{\Psi_k}{T_k}+2) \prod_{j=1}^{k} (\frac{U_j}{M} + 1)\leq 3
\end{equation}
or
\begin{equation}
\label{eq:schedulability-DAG-augmentation-GRM}
\sum_{j=1}^{k} \frac{U_j}{M}\leq \ln{\frac{3}{\frac{\Psi_k}{T_k}+2}}.
\end{equation}
\end{theorem}
\begin{proof}
Based on Lemma~\ref{lemma:gRM-DAG-sufficient}, which is very similar
to Lemma~\ref{lemma:gRM-sufficient}, we can perform a
similar transformation as in Theorem~\ref{thm:multiprocessor-GRM}, in
which $\alpha_i \leq
  \frac{2}{M}$ and $\beta_i \leq \frac{1}{M}$. By adopting
  Lemma~\ref{lemma:framework-constrained}, we know that if 
\begin{equation}
\label{eq:schedulability-DAG-weaker}
 (\frac{\Psi_k + \frac{C_k-\Psi_k}{M}}{T_k}+2) \prod_{j=1}^{k-1} (\frac{U_j}{M} + 1)\leq 3,
\end{equation}
then task $\tau_k$ is schedulable. Due to the fact that $C_k-\Psi_k
\leq C_k$, we know that $(\frac{\Psi_k + \frac{C_k-\Psi_k}{M}}{T_k}+2)
\leq (\frac{\Psi_k + \frac{C_k}{M}}{T_k}+2) \leq
(\frac{\Psi_k}{T_k}+2) \cdot (\frac{U_k}{M} + 1)$. Therefore, if the
condition in Eq.~\eqref{eq:schedulability-DAG} holds, the condition in
Eq.~\eqref{eq:schedulability-DAG-weaker} also holds, which implies the
schedulability. With the result in Eq.~\eqref{eq:schedulability-DAG},
we can use the same procedure as in
Lemma~\ref{lemma:framework-totalU-exclusive} to obtain
Eq.~\eqref{eq:schedulability-DAG-augmentation-GRM}.
\end{proof}

\begin{corollary}
  The capacity augmentation factor of global RM for a sporadic DAG system with implicit deadlines is $3.62143$.
\end{corollary}
\begin{proof}
  Suppose that $\sum_{\tau_i} \frac{C_i}{M T_i} \leq \frac{1}{b}$ and
  $\frac{\Psi_k}{T_k} \leq \max_{\tau_i} \frac{\Psi_i}{T_i} \leq
  \frac{1}{b}$.  Therefore, by
  Eq.~\eqref{eq:schedulability-DAG-augmentation-GRM}, we can guarantee
  the schedulability of task $\tau_k$ if $\frac{1}{b} \leq
  \ln{\frac{3}{2+\frac{1}{b}}}$. This is equivalent to solving $x =
  \ln{\frac{3}{2+x}}$, which holds when $x\approx
  3.62143$ by solving the equation numerically. Therefore,
  we reach the conclusion of the capacity augmentation factor
  $3.62143$.
\end{proof}

\subsection{Global RM for Self-Suspending Tasks}
\label{sec:global-SSS}

The self-suspending task model extends the sporadic task model by allowing tasks to suspend themselves.
An overview of work on scheduling self-suspending task systems can be
found in \cite{RTSS14a}.
\ifbool{techreport}{In \cite{RTSS14a}, a general interference-based analysis framework was developed that can be applied to derive sufficient utilization-based tests for self-suspending task systems on uniprocessors.}{}

Similar to sporadic tasks, a  self-suspending task releases jobs sporadically. Jobs alternate between computation and suspension phases. We assume that each job of $\tau_i$ executes for at most $C_i$ time units (across all of its execution phases) and suspends for at most $S_i$ time units (across all of its suspension phases). We assume that $C_i+S_i \leq T_i$ for any task $\tau_i \in \tau$; for otherwise deadlines would be missed. The self-suspending model is general: we place no restrictions on the number of phases per-job and how these phases interleave (a job can even begin or end with a suspension phase). Different jobs belong to the same task can also have different phase-interleaving patterns.
For many applications, such a general self-suspending model is needed due to the unpredictable nature of I/O operations.
We have the following lemma, in which the proof is in Appendix B\citetechreport{}.
\begin{lemma}
  \label{lemma:gRM-suspension-sufficient}
  Task $\tau_k$ in a  self-suspending system with implicit deadlines  is schedulable under global RM on $M$ identical
  processors, if
  \begin{equation}
    \label{eq:gRM-suspension-sufficient}
\exists t \mbox{ with } 0 < t \leq T_k {\;\; and \;\;} C_k+S_k +\sum_{i=1}^{k-1} 
\frac{W_i(t)}{M} \leq t, 
  \end{equation}
  where $W_i(t)$ is defined in Eq.~\eqref{eq:W_i-multiprocessor}.
\end{lemma}

\begin{theorem}
\label{thm:multiprocessor-suspension}
Task $\tau_k$ in a sporadic self-suspending system with implicit deadlines  is schedulable by global RM on
$M$ processors if
\begin{equation}
\label{eq:schedulability-suspension}
(\frac{C_k+S_k}{T_k}+2) \prod_{j=1}^{k-1} (\frac{U_j}{M} + 1)\leq 3.
\end{equation}
\end{theorem}
\begin{proof}
By Lemma~\ref{lemma:gRM-suspension-sufficient}, we can perform a
similar transformation as in Theorem~\ref{thm:multiprocessor-GRM} with $\alpha_i \leq
  \frac{2}{M}$ and $\beta_i \leq \frac{1}{M}$.
\end{proof}







\vspace{-0.05in}
\section{Conclusion}
\vspace{-0.05in}

With the presented applications, we believe that the general schedulability analysis \framework{}  framework for fixed-priority scheduling has high potential to be adopted for analyzing other task models in real-time systems. We constrain ourselves by demonstrating the applications for simple scheduling policies, like global RM in multiprocessor scheduling. The framework can be used, once the $k$-point effective scheduling test can be constructed. Although the emphasis of this paper is not to show that the resulting tests for different task models by applying the \framework{} framework are better than existing work, some analysis results by applying the  \framework{}  framework have been shown superior to the state of the art. 
For completeness, another document has been prepared in
\cite{DBLP:journals/corr/framework-compare} to present the similarity,
the difference and the characteristics of \framework{} and
\frameworkkq{}. With our frameworks, some difficult
schedulability test and response time analysis problems may be solved
by building a good (or exact) exponential-time test and applying these
frameworks.

Appendix D\citetechreport{} provides some case studies with evaluation
results of some selected utilization-based schedulability
tests. Appendix E\citetechreport{} further provides some additional
properties that come directly from the \framework{} framework.  More
applications can be found in partitioned scheduling \cite{DBLP:journals/corr/Chen15k}, non-preemptive
scheduling \cite{DBLP:conf/ecrts/BruggenCH15}, etc.

\begin{spacing}{0.9}
\noindent{\small 
{\bf Acknowledgement}: This paper has been supported by DFG, as
    part of the Collaborative Research Center SFB876
    (http://sfb876.tu-dortmund.de/), and the priority program
    "Dependable Embedded Systems" (SPP 1500 -
    http://spp1500.itec.kit.edu). We would also like to thank Dr. Vincenzo
    Bonifaci for his valuable input to improve the presentation of the paper.
}
\end{spacing}

\footnotesize

\vspace{-0.05in}
\ifbool{techreport}{
\begin{spacing}{0.98}
}{
\begin{spacing}{0.85}
\vspace{-0.03in}
}
\def\IEEEbibitemsep{-0.3pt}
\bibliographystyle{abbrv}
\bibliography{ref,real-time}
\end{spacing}

\normalsize

\vspace{-0.15in}
\section*{Appendix A: Monotonic Schedulability Test}

The tests presented in the theorems or corollaries do not guarantee to have the monotonicity with respect to the $k$-th highest-priority task.  However, by sacrificing the quality of the schedulability tests, we can still obtain monotonicity, with which the schedulability test of a task set can be done with linear-time complexity. These tests can be used for on-line admission control.  For example, the test in Theorem \ref{thm:multiprocessor-DAG} can be modified to the following theorem:

\begin{theorem}
\label{thm:gRM-DAG-fast}
An implicit-deadline DAG system $\tau$ is schedulable by global RM on
$M$ processors if
\begin{equation}
 (\Delta_{\max}+2) \prod_{\tau_i} (\frac{U_i}{M} + 1)\leq 3,
\end{equation}
where $\Delta_{\max}$ is $max_{\tau_i \in \tau} \frac{\Psi_i}{T_i}$.
\end{theorem}

\vspace{-0.1in}
\section*{Appendix B: Proofs}


\begin{appProof}{Lemma~\ref{lemma:framework-totalU-constrained}}
  This lemma is proved by sketch with Lagrange Multiplier to find the infimum $\frac{C_k}{t_k}+\sum_{i=1}^{k-1} U_i$ such that Eq.~\eqref{eq:schedulability-constrained} does not hold, which is a non-linear programming problem. Due to the fact that  $(1+\beta U_1)(1+\beta U_2) \leq (1+\beta\frac{U_1+U_2}{2})^2$ when $\beta \geq 0, U_1 \geq 0, U_2 \geq 0$,  the infimum 
$\frac{C_k}{t_k}+\sum_{i=1}^{k-1} U_i$ happens when $U_1=U_2=\cdots=U_{k-1}$. So, there are only two variables $\frac{C_k}{t_k}$ and $U_1$ to minimize $\frac{C_k}{t_k}+(k-1)U_1$ such that $(\frac{C_k}{t_k}+\frac{\alpha}{\beta})(\beta U_1+1)^{k-1} \geq \frac{\alpha}{\beta}+1$.

Let $\lambda$ be the Lagrange Multiplier and $F$ be $\frac{C_k}{t_k}+(k-1)U_1-\lambda\left((\frac{C_k}{t_k}+\frac{\alpha}{\beta})(\beta U_1+1)^{k-1} -( \frac{\alpha}{\beta}+1)\right)$.
The minimum $\frac{C_k}{t_k}+(k-1)U_1$ happens when
$\frac{\partial F}{\partial U_1}=(k-1)-\lambda\beta(k-1) (\frac{C_k}{t_k}+\frac{\alpha}{\beta}) (\beta U_1+1)^{k-2}=0$ and $\frac{\partial  F}{\partial \frac{C_k}{t_k}}=1-\lambda(\beta U_1+1)^{k-1}=0$.
When $k \geq 2$, by reorganizing the above two equations, we have
$1-\frac{\beta(\frac{\alpha}{\beta} + \frac{C_k}{t_k})}{\beta U_1+1} = 0$.
By the Lagrange Multiplier method, the minimum happens when $(\frac{C_k}{t_k}+\frac{\alpha}{\beta})(\beta U_1+1)^{k-1} = \frac{\alpha}{\beta}+1$ and 
$1-\frac{\beta(\frac{\alpha}{\beta} + \frac{C_k}{t_k})}{\beta U_1+1} = 0$. By solving the above equation, 
the non-linear programming is minimized when $U_1$ is $\frac{(\alpha+\beta)^{\frac{1}{k}}-1}{\beta}$ and $\frac{C_k}{t_k}$ is 
$\frac{(\alpha+\beta)^{\frac{1}{k}}-\alpha}{\beta}$.

We also need to consider the boundary cases when $(\alpha+\beta)^{\frac{1}{k}}-1 < 0$ or $(\alpha+\beta)^{\frac{1}{k}}-\alpha < 0$ with Karush Kuhn Tucker (KKT) conditions.
The Lagrange Multiplier method may result in a solution with a negative $U_1$ when $(\alpha+\beta)^{\frac{1}{k}}-1 < 0$. If this happens, we know that the extreme case happens when $U_1$ is $0$ by using KKT condition. Moreover, if $(\alpha+\beta)^{\frac{1}{k}}-\alpha < 0$, then we know that $\frac{C_k}{t_k}$ should be set to $0$ in the extreme case by using KKT condition.
By the above analysis, we reach the conclusion in Eq.~\eqref{eq:schedulability-totalU-constrained}.
 \end{appProof}

\begin{appProof}{Lemma~\ref{lemma:framework-totalU-exclusive}}
  This comes directly from Eq.~\eqref{eq:schedulability-constrained}
  in Lemma~\ref{lemma:framework-constrained} with a simpler Lagrange
  Multiplier procedure as in the proof of
  Lemma~\ref{lemma:framework-totalU-constrained}, in which the infimum
  total utilization under
 $
\prod_{j=1}^{k-1} (\beta U_j + 1) >  \frac{\frac{\alpha}{\beta}+1}{\frac{C_k}{t_k}+\frac{\alpha}{\beta}}
  $
  happens when all the $k-1$ tasks have the same
  utilization.   
\end{appProof}

\begin{appProof}{Lemma~\ref{lemma:framework-general}}
  The first part of the proof by constructing the corresponding linear programming to minimize $C_k^*$ as follows is the same as in the proof of Lemma~\ref{lemma:framework-constrained} by setting $t_k^*$ as $t_k +s$ with $s \geq 0$:
\begin{subequations}
    \label{eq:lp-framework-general-vary-alpha}
  \begin{align}
    \mbox{min } &  
    t_k^* - \sum_{i=1}^{k-1}(\alpha_i +\beta_i) U_it_i^*\\
    \mbox{s.t.} \;\;&
t_k^*- \sum_{i=j}^{k-1} \beta_i t_i^* U_i \geq t_j^*
      & \forall 1 \leq j \leq k - 1, \label{eq:lp-framework-general-vary-alpha-constraints}\\
& t_j^* \geq 0
      & \forall 1 \leq j \leq k - 1. \label{eq:lp-framework-general-vary-alpha-boundaryconstraints}\\
&     t_k^* \geq t_k.
  \end{align}    
  \end{subequations}
Similarly, a feasible extreme point solution can be represented by two sets ${\bf T}_1$ and ${\bf T}_2$ of the $k-1$ higher-priority tasks, in which $t_j^* = 0$ if $\tau_j$ is in ${\bf T}_1$ and $t_j^* >0 $ if task $\tau_j$ is in ${\bf T}_2$.

One specific extreme point solution is to have ${\bf T}_1 = \emptyset$. For such a case, we use the same steps from Eq.~\eqref{eq:periodrelation} to Eq.~\eqref{eq:3rdperiodrelation}:
\begin{equation}
\label{eq:3rdperiodrelation-general}
\dfrac{t_i^*}{t_k^*}= \prod_{j=i}^{k-1}\frac{t_j^*}{t_{j+1}^*} = \frac{1}{\prod_{j=i}^{k-1} (\beta_j U_j + 1)}.
\end{equation}
The resulting objective function of this extreme point solution for Eq.~\eqref{eq:lp-framework-general-vary-alpha} is $ t_k^*(1-\sum_{i=1}^{k-1}  \frac{U_i(\alpha_i
  +\beta_i)}{\prod_{j=i}^{k-1} (\beta_jU_j + 1)})$.
The above steps are identical to Step 1 and Step 2 in the proof of Lemma~\ref{lemma:framework-constrained}. 
We will show, similarly to Step 3 in the proof of Lemma~\ref{lemma:framework-constrained}, for the rest of the proof, that the above extreme point solution is either optimal for the objective function of Eq.~\eqref{eq:lp-framework-general-vary-alpha} or $ 1-\sum_{i=1}^{k-1}  \frac{U_i(\alpha_i
  +\beta_i)}{\prod_{j=i}^{k-1} (\beta_jU_j + 1)} \leq 0$.

 For a feasible extreme point solution with ${\bf T}_1\neq \emptyset$, we will convert it to the above extreme point solution with ${\bf T}_1 = \emptyset$ by steps, in which each step moves one task from ${\bf T}_1$ to ${\bf T}_2$ by decreasing the objective function in the linear programming in Eq.~\eqref{eq:lp-framework-general-vary-alpha}.  For the rest of the proof, we start from a feasible extreme point solution, specified by $S=<{\bf T}_1, {\bf T}_2>$. Suppose that $\tau_\ell$ is the \emph{first} task in this extreme point solution $S$ with $t_\ell^*$ set to $0$, i.e., $t_k^*- \sum_{i=j}^{k-1} \beta_i t_i^* U_i = t_j^*$ for $j=1,2,\ldots,\ell-1$.

Assume that $\omega > \ell$ is the index of the next task with $t_{\omega}^* > 0$ in the extreme point solution $S$, i.e., $t_{\ell}^* = t_{\ell+1}^* = \cdots= t_{\omega-1}^*=0$. If all the remaining tasks are with $t_i^*=0$ for $\ell \leq i \leq k-1$, then $\omega$ is set to $k$ and $t_{\omega}^*$ is $t_k^*$. 
If $\ell$ is $1$, we can easily set $t_1^*$ to $\frac{t_{\omega}^*}{1+\beta_1 U_1}$, which is $> 0$ and the objective function of the linear programming becomes smaller. We focus on the cases where $\ell > 1$.
We can use the same steps from Eq.~\eqref{eq:periodrelation} to Eq.~\eqref{eq:3rdperiodrelation}:
$t_{i+1}^*-t_i^* = \beta_i U_i$ for $i=1,2,\ldots,\ell-2$ and $t_{\omega}^*-t_{\ell-1}^* = \beta_{\ell-1}U_{\ell-1}$. Therefore, $\sum_{i=1}^{\ell-1} (\alpha_i + \beta_i)U_i t_i^* = t_{\omega}^*(\sum_{i=1}^{\ell-1}  \frac{U_i(\alpha_i
  +\beta_i)}{\prod_{j=i}^{\ell-1} (\beta_jU_j + 1)})$. There are two cases:

{\bf Case 1:} If $\sum_{i=1}^{\ell-1}  \frac{U_i(\alpha_i
  +\beta_i)}{\prod_{j=i}^{\ell-1} (\beta_jU_j + 1)} \geq 1$, then we can conclude
{\footnotesize \begin{align*}
  &\sum_{i=1}^{k-1}  \frac{U_i(\alpha_i
  +\beta_i)}{\prod_{j=i}^{k-1} (\beta_jU_j + 1)} \\
= \;\; &
  \sum_{i=1}^{\ell-1}  \frac{U_i(\alpha_i
  +\beta_i)}{(\prod_{j=i}^{\ell-1} (\beta_jU_j + 1)) (\prod_{j=\ell}^{k-1} (\beta_jU_j + 1))} + 
  (\sum_{i=\ell}^{k-1}  \frac{U_i(\alpha_i
  +\beta_i)}{\prod_{j=i}^{k-1} (\beta_jU_j + 1)})\\
\overset{1}{\geq}\;  &
  \frac{1}{\prod_{j=\ell}^{k-1} (\beta_jU_j + 1)} + 
  (\sum_{i=\ell}^{k-1}  \frac{U_i \beta_i}{\prod_{j=i}^{k-1} (\beta_jU_j + 1)})\\
= \;\;&   \frac{1}{\prod_{j=\ell}^{k-1} (\beta_jU_j + 1)} + (1 - 
\frac{1}{\prod_{j=\ell}^{k-1} (\beta_jU_j + 1)})\\
= \;\;& 1,
\end{align*}
}
where $\overset{1}{\geq}$ comes from the assumption $\sum_{i=1}^{\ell-1}  \frac{U_i(\alpha_i
  +\beta_i)}{\prod_{j=i}^{\ell-1} (\beta_jU_j + 1)} \geq 1$ and $\alpha_\ell > 0$.

{\bf Case 2:} If $\sum_{i=1}^{\ell-1}  \frac{U_i(\alpha_i
  +\beta_i)}{\prod_{j=i}^{\ell-1} (\beta_jU_j + 1)} <1$, then we can greedily set $t_{\ell}^* > 0$ (i.e., move task $\tau_\ell$ from ${\bf T}_1$ to ${\bf T}_2$). Such a change of $\tau_\ell$ from ${\bf T}_1$ to ${\bf T}_2$ has no impact on task $\tau_i$ with $i > \ell$, but has impact on all the tasks $\tau_i$ with $i \leq \ell$. 
That is, after changing, by 
 using the same steps from Eq.~\eqref{eq:periodrelation} to Eq.~\eqref{eq:3rdperiodrelation}, we have
$t_{i+1}^*-t_i^* = \beta_i U_i$ for $i=1,2,\ldots,\ell-1$ and $t_{\omega}^*-t_{\ell}^* = \beta_{\ell}U_{\ell}$.
The change of the objective function in Eq.~\eqref{eq:lp-framework-general-vary-alpha} after moving task $\tau_\ell$  from ${\bf T}_1$ to ${\bf T}_2$ is
{\footnotesize 
\begin{align*}
&t_{\omega}^*\left(-\frac{U_{\ell}(\alpha_\ell+\beta_{\ell}) +  \sum_{i=1}^{\ell-1}  \frac{U_i(\alpha_i
  +\beta_i)}{\prod_{j=i}^{\ell-1} (\beta_jU_j + 1)}}{1+\beta_{\ell}U_{\ell} } + \sum_{i=1}^{\ell-1}  \frac{U_i(\alpha_i
  +\beta_i)}{\prod_{j=i}^{\ell-1} (\beta_jU_j + 1)}\right)\\
&=\;\;\; -t_{\omega}^*\left(\frac{U_{\ell}(\alpha_\ell+\beta_{\ell}) -  \beta_\ell U_\ell\sum_{i=1}^{\ell-1}  \frac{U_i(\alpha_i
  +\beta_i)}{\prod_{j=i}^{\ell-1} (\beta_jU_j + 1)}}{1+\beta_{\ell}U_{\ell} } \right)\\
&\overset{2}{<}\;\;  -t_{\omega}^*\left(\frac{U_{\ell}(\alpha_\ell+\beta_{\ell}) -  \beta_\ell U_\ell}{1+\beta_{\ell}U_{\ell} } \right) < 0,
\end{align*}
}
where $\overset{2}{<}$  comes from the condition $\sum_{i=1}^{\ell-1}  \frac{U_i(\alpha_i
  +\beta_i)}{\prod_{j=i}^{\ell-1} (\beta_jU_j + 1)} <1$.

With the above two cases, either (1) we can repeatedly move one task
from ${\bf T}_1$ to ${\bf T}_2$ by changing the extreme point solution
$S$ to another extreme point solution $S'$ to improve the objective
function in Eq.~\eqref{eq:lp-framework-general-vary-alpha} or (2) $
1-\sum_{i=1}^{k-1} \frac{U_i(\alpha_i +\beta_i)}{\prod_{j=i}^{k-1}
  (\beta_jU_j + 1)} \leq 0$. Similar to the argument in the proof of Lemma~\ref{lemma:framework-constrained}, the minimum $C_k^*$ happens when $s$ is $0$ if $
1-\sum_{i=1}^{k-1} \frac{U_i(\alpha_i +\beta_i)}{\prod_{j=i}^{k-1}
  (\beta_jU_j + 1)} > 0$.
\end{appProof}

\begin{appProof}{Theorem~\ref{thm:speedup-DM}}
Clearly, if $\prod_{j=1}^{k-1} (U_j +
  1) \geq 2$, we can already conclude that $\sum_{j=1}^{k-1} U_j \geq
  \ln{2}$ by following the same analysis in \cite{liu1973scheduling,bini2003rate}, and the speed-up factor is $1/\ln{2} < 1.76322$ for such a case. We
  focus on the other case with $\prod_{j=1}^{k-1} (U_j +
  1) < 2$. 

To understand whether the task set is schedulable under any scheduling
policy, we only have to test the feasibility of preemptive EDF
schedule, as preemptive EDF is an optimal scheduling policy to meet
the deadlines in uniprocessor systems. Baruah et al. \cite{DBLP:conf/rtss/BaruahMR90} provide
a demand-bound function (dbf) test to verify such a case. That is,
the demand bound function $dbf_i(t)$ of task $\tau_i$ with interval length $t$ is
\[
dbf_i(t) = \max\left\{0,
  \floor{\frac{t-D_i}{T_i}}+1\right\}C_i.
\]
A system of independent, preemptable, sporadic tasks can be feasibly
scheduled (under EDF) on a processor if and only if
\[
\forall t \geq 0\;\; \sum_{i} dbf_i(t) \leq t.
\]
Therefore, if there exists $t$ such that $\frac{\sum_{i} dbf_i(t)}{s} > t$ or $\sum_{i} U_i > s$, then the task set is not schedulable by EDF on a uniprocessor platform with speed $s$.

  Recall that we can construct the corresponding $k$-point effective schedulability test defined in Definition~\ref{def:kpoints} with $\alpha_i=1$ and $\beta_i \leq 1$ as shown in the proof of Theorem~\ref{theorem:sporadic-general}. 
 Now, we take a look of the proof in Lemma~\ref{lemma:framework-constrained} again. The same proof can also be applied to show that the extreme point solution that leads to the solution in Eq.~\eqref{eq:alpha+beta} is also an optimal solution for the following linear programming when $\prod_{j=1}^{k-1} (U_j +
  1) < 2$:
  \begin{subequations}
\small  \begin{align*}
    \mbox{infimum \;\;} & C_k^* + \sum_{i=1}^{k-1}  t_i^* U_i\\
    \mbox{s.t.\;\;} &     C_k^* + \sum_{i=1}^{k-1}  t_i^* U_i + \sum_{i=1}^{j-1}  t_i^* U_i > t_j^* \geq 0,\;\forall j=1,2,\ldots, k.
       \end{align*}    
  \end{subequations}
  That is, the corresponding objective function of Eq.~\eqref{eq:lp-framework-constrained} is $t_k+ s-  \beta\sum_{i=1}^{k-1} U_it_i^*$, by setting $\alpha=1$ and $\beta=1$.
  Let $C_k^*$ be the optimal $C_k^*$ of the above linear programming when $\prod_{j=1}^{k-1} (U_j +
  1) < 2$. By the above argument, we know that $(1+\frac{C_k^*}{D_k})\prod_{j=1}^{k-1} (U_j +
  1) = 2$. Therefore, $C_k^* > 0$.  Moreover,
  \begin{equation}
    \label{eq:extreme-constrained-utilization}
    \prod_{j=1}^{k-1} (U_j + 1) = \frac{2}{1+\frac{C_k^*}{D_k}}.
  \end{equation}

  For the rest of the proof, let $C_k^*/D_k$ be $x$. 
 If task $\tau_k$ is not
  schedulable by DM (or does not pass Eq.~\eqref{eq:schedulability-sporadic-any-a}),  then,
  \begin{align*}
    &\frac{C_k'+
      \sum_{i=1}^{k-1} dbf_i(D_k)}{D_k} \geq \frac{C_k' +
      \sum_{i=1}^{k-1} t_i U_i}{D_k}  \\
> \;\;&  \frac{C_k^* + \sum_{i=1}^{k-1}
  t_i^* U_i}{t_k} 
= \frac{t_1^*}{t_k}\\
  \overset{1}{=}\;\; &\frac{1}{\prod_{j=1}^{k-1} (U_j + 1)}
  = \frac{1+\frac{C_k^*}{D_k}}{2} = \frac{1+x}{2}
  \end{align*}
  where $\overset{1}{=}$ comes from the relation $\frac{t_1^*}{t_k}$ in
  Eq.~(\ref{eq:3rdperiodrelation}) when $s$ is $0$.
  Moreover, with Lemma~\ref{lemma:framework-totalU-exclusive}, we have
  \begin{equation}
    \label{eq:DM-utilization}
    \sum_{i=1}^{k-1} U_i >\ln(\frac{2}{1+x}).
  \end{equation}
  Due to the fact that $\frac{1+x}{2}$ is an increasing function of
  $x$ and $\ln(\frac{2}{1+x})$ is a decreasing function of x, we know
  that $\inf_{0 \leq x < 1}\max\left\{\frac{1+x}{2}, \ln(\frac{2}{1+x})\right\}$ is the intersection  
of $\frac{1+x}{2}$ and $\ln(\frac{2}{1+x})$, which is
  $1/1.76322$. Therefore, 
  \begin{align*}
& \max\left\{ \frac{C_k'+ \sum_{i=1}^{k-1}
    dbf_i(D_k)}{D_k}, \sum_{i=1}^{k-1} U_i\right\}\\
> &   \max\left\{\frac{1+x}{2}, \ln(\frac{2}{1+x})\right\} \geq \frac{1}{1.76322}.
  \end{align*}
 As a result, the speed-up factor of the test in Eq.~\eqref{eq:schedulability-sporadic-any-a} for DM scheduling for constrained-deadline systems is $1.76322$.
\end{appProof}

\ifbool{techreport}{}{\end{document}}
\begin{appProof}{Lemma \ref{lemma:gRM-DAG-sufficient}}
  This is based on the simple observations in the previous results,
  e.g.,
  \cite{DBLP:conf/rtss/Baker03,DBLP:conf/rtss/GuanSYY09,Li:ECRTS14}.
  We prove by contrapositive. Suppose that a job of task $\tau_k$
  misses its deadline. Let the arrival time of this job be $a$ and the
  absolute deadline be $a+D_k$.  Let $X$ be the total amount of time
  in $(a, a+D_k]$, in which at least one processor is not executing
  any job. Due to the assumption that $\tau_k$ misses its deadline,
  the DAG structure of task $\tau_k$, and the global RM scheduling
  policy, we know that $X \leq \Psi_k$. The workload resulting from
  the higher-priority tasks in $(a, a+t]$ is at most $W_i(t)$, by
  greedily considering that the job of $\tau_i$ released before $a$ is
  completely not executed before $a$. This part is pessimistic enough
  to be independent upon the DAG structure. Therefore, we know that
  the unschedulability of task $\tau_k$ implies that
\begin{equation}
\forall t \mbox{ with } 0 < t \leq T_k {\;\; and \;\;} \Psi_k+\frac{C_k-\Psi_k}{M}+
\sum_{i=1}^{k-1} 
\frac{W_i(t)}{M}> t,     
  \end{equation}
  which concludes the proof.
\end{appProof}

\begin{appProof}{Lemma~\ref{lemma:gRM-suspension-sufficient}}
This lemma can be proved in a similar manner as shown in our previous work~\cite{RTSS14a}. We prove by contrapositive. Suppose that a job of task $\tau_k$, $\tau_{k,j}$, misses its deadline. Let the arrival time of this job be $a$ and the absolute deadline be $a+D_k$. 

We first construct a task set $\tau'$ from $\tau$, where the only difference between the two task sets is on $\tau_{k}$. In $\tau'$, we convert all suspensions of jobs released by $\tau_k$ into computation. That is, we treat $\tau_k$ as an ordinary sporadic task by factoring its suspension length into the worst-case execution time parameter. Thus, $\tau_k$ executes just like an ordinary sporadic task (without suspensions) in the corresponding schedule, with an execution time of $C_k+S_k$. Note that $\tau_k$'s computation (both the original computation and the computation converted from suspensions) will be preempted by higher-priority tasks. If $\tau_{k,j}$ in the original task set $\tau$ is not schedulable, then in the interval $(a, a+D_k]$, the system can idle or execute tasks with lower priority than $\tau_k$ by at most $S_k$ amount of time; otherwise, job $\tau_{k, j}$ has to suspend more than $S_k$ amount of time in this interval. In the setting of $\tau'$, we can consider the same pattern for the other jobs, but only convert the suspensions of task $\tau_k$ in $\tau$ to computation time. The additional $S_k$ amount of computation time of $\tau_{k, j}$ in $\tau'$ can only be granted when the processor is idle or executes tasks with lower priority than $\tau_k$, which is in total at most $S_k$ as explained above. Therefore, $\tau_k$ in $\tau'$ is also not schedulable under global RM.  

Within $(a, a+t_o] \in (a, a + D_k]$, the work done by any high-priority task $\tau_i$ ($i < k$) in the worst case can be divided into three parts: (\textit{i}) body jobs: jobs of $\tau_i$ with both release time and absolute deadline in $(a, a+t_o]$, (\textit{ii}) carry-in job: a job of $\tau_i$ with release time earlier than $a$ and absolute deadline in $(a, a+t_o]$, and (\textit{iii}) carry-out job: a job of $\tau_i$ with release time in $(a, a+t_o]$ and absolute deadline after $a+t_o$. Since the carry-in and the carry-out job can each contribute at most $C_i$ workload in $[a, a+t_o]$, a safe upper bound of the interference due to task $\tau_i$ in $(a, a+t_o]$ is obtained by assuming that the carry-in and carry-out jobs of $\tau_i$ both contribute $C_i$ each in $(a, a+t_o]$. Thus, the workload resulting from any higher-priority task $\tau_i$ in $(a, a+t_o]$ is at most $W_i(t_o)$ (defined in Eq.~\eqref{eq:W_i-multiprocessor}). Therefore, in order for $\tau_{k,j}$ in $\tau'$ to miss its deadline at $a+D_k$, we know that
  \begin{equation}
\forall t \mbox{ with } 0 < t \leq T_k {\;\; and \;\;} C_k+S_k +\sum_{i=1}^{k-1} 
\frac{W_i(t)}{M} > t, \nonumber
  \end{equation}
must hold, which concludes the proof.
\end{appProof}

\section*{Appendix C: Application for Multi-frame Tasks}

This section adopts the schedulability test framework in
Sec.~\ref{sec:framework} for multi-frame real-time tasks, proposed by Mok and Chen \cite{DBLP:dblp_journals/tse/MokC97}.
A multi-frame real-time task $\tau_i$ with $m_i$ frames is defined as a
sporadic task with period $T_i$ with an array
$C_{i,0}, C_{i,1}, \ldots, C_{i,m_i-1}$ of different execution
times. The execution time of the $j$-th job of task
$\tau_i$ is defined as $C_{i,(j\mod m_i)}$.

Mok and Chen~\cite{DBLP:dblp_journals/tse/MokC97} propose a utilization-based schedulability under \emph{rate monotonic} (RM) scheduling by generalizing the Liu \& Layland bound~\cite{liu1973scheduling} for the multiframe task.
Kuo et al. ~\cite{kuo2003efficient} present a more precise scheduability test by merging the tasks with harmonic periods before inspecting the Mok~\&~Chen bound.
The researches in \cite{journals/tc/LeeSP04,DBLP:conf/rtas/WuLZ05} also demonstrate how to apply their methods to handle the multi-frame task model.
Lu et al.~\cite{lu2007new} further consider the ratio between periods to improve the existing utilziation-based test.

For a multi-frame task, we define the utilization $U_i$ of task $\tau_i$ 
based on its peak utilization, i.e.,
$U_i=\frac{\max_{j=0}^{m_i-1}C_{i,j}}{T_i}$.
 Without loss of generality, we assume that each task has at least two frames,
i.e., $m_i \geq 2$. If a task has only one frame, we can artificially
create a corresponding multi-frame task with $2$ frames and with the
same execution time.
We will limit our attention in uniprocessor systems.

Let $\phi_i(\ell)$ be the maximum of the sum of the execution time of
any $\ell$ consecutive frames of task $\tau_i$. It is clear that
$\phi_i(1)$ is $\max_{j=0}^{m_i-1}C_{i,j}$ and $\phi_i(2)$ is
$\max_{j=0}^{m_i-1} (C_{i,j} + C_{i,((j+1) \mod m_i)})$. Therefore, we know
that $U_i$ is equal to $\frac{\phi_i(1)}{T_i}$. For brevity, we define
$\phi_i(0)$ as $0$.
It is not difficult to see that $\phi_i(\ell)$ is equal to
$\phi_i(\ell \mod m_i) +
\floor{\frac{\ell}{m_i}}\sum_{j=0}^{m_i-1}C_{i,j}$ when $\ell > m_i$,
where $\phi_i(0)$ is set to $0$ for notational brevity. Therefore, we
only need to build a table for the first $m_i$ entries to construct
$\phi_i(\ell)$.  Deriving $\phi_i(\ell)$ can be done in $O(m_i^2)$ for
$\ell=1,2,\ldots,m_i-1$.


Again, we consider testing the schedulability of task $\tau_k$ under
RM scheduling, in which there are $k-1$ higher-priority
multi-frame tasks $\tau_1, \tau_2, \ldots, \tau_{k-1}$.  We have the
following schedulability condition for RM.
\begin{lemma}
  \label{lemma:multi-frame-sufficient}
  Suppose that all the multi-frame tasks with higher priority than
  $\tau_k$, i.e., $\tau_1, \tau_2, \ldots, \tau_{k-1}$, are
  schedulable by RM. Multi-frame task $\tau_k$ is schedulable under RM
  on a uniprocessor, if
  \begin{equation}
    \label{eq:multiframe-sufficient}
\exists t \mbox{ with } 0 < t \leq T_k {\;\; and \;\;} \phi_k(1) +\sum_{i=1}^{k-1} 
\phi_i(\ceiling{\frac{t}{T_i}})\leq t.
  \end{equation}
\end{lemma}
\begin{proof}
  This comes from Theorem 5 and Lemma 6 by Mok and Chen in
  \cite{DBLP:dblp_journals/tse/MokC97}.
\end{proof}
We now present a more pessimistic schedulability test than
Eq.~\eqref{eq:multiframe-sufficient} to construct a $k$-point
effective schedulability test.  Let $\delta_i(j)$ be
$\phi_i(j)-\phi_i(j-1)$. That is, $\delta_i(j)$ is the additional
workload released from the $j$-th invocation of task $\tau_i$ in the
definition of $\phi_i()$. Moreover, let $\delta_i^{\min}(\ell)$ be
$\min_{j=1,2,\ldots,\ell} \delta_i(j)$, i.e., $\delta_i^{\min}(\ell)$
is the minimum $\delta_i(j)$ among the first $\ell$ release of task
$\tau_i$. 

We further define $\phi_i'(\ell)$ as follows:
\begin{equation}
\label{eq:phi-upper}
  \phi_i'(\ell) = 
   \phi_i(\ell+1)-\delta_i^{min}(\ell+1)
\end{equation}
The definition of $\phi_i'(\ell)$ comes from the operation by swapping
the increased workload of the $\ell+1$-th release of task $\tau_i$
with the workload $\delta_i^{min}(\ell+1)$.

Again, let $t_i$ be $\floor{\frac{T_k}{T_i}}T_i$ for $i=1,2,\ldots,k$,
and reindex the tasks such that $t_1 \leq t_2 \leq \ldots \leq t_k$.
Instead of testing all the $t$ values in
Eq.~\eqref{eq:multiframe-sufficient} by referring to $\phi_i()$, we
only apply the test for these $k$ different $t_i$ values by referring
to $\phi_i'()$ as shown in the following lemma.

\begin{lemma}
  \label{lemma:multi-frame-k-point-sufficient}
Multi-frame task $\tau_k$ is schedulable under RM
  on a uniprocessor, if there exists $t_j$ such that
\begin{equation}
  \label{eq:multiframe-general-sufficient-k-points-2}
  \phi_k'(1)
  + \sum_{i=1}^{k-1} \phi_i'\left(\frac{t_i}{T_i}\right) + \sum_{i=1}^{j-1}
  \delta_i^{min}\left(\frac{t_i}{T_i}+1\right)\leq t_j,
\end{equation}
where $t_i$ is $\floor{\frac{T_k}{T_i}}T_i$ for $i=1,2,\ldots,k$ and
$t_i \leq t_{i+1}$.
\end{lemma}
\begin{proof}
   This property is due to the fact that $\phi_i'(\ell) \geq
   \phi_i(\ell)$ for any $0 \leq \ell \leq \ceiling{\frac{T_k}{T_i}}$
   according to the definition of $\phi_i'()$.
\end{proof}

\begin{theorem}
\label{thm:multiframe}
Let $f_i$ be
$\frac{\phi_i(1) \cdot \ell_i}{\delta_i^{\min}(\ell_i+1)}$, where
$\ell_i$ is defined as $\frac{t_i}{T_i}=\floor{{\frac{T_k}{T_i}}}$ for notational brevity.
Task $\tau_k$ is schedulable under RM if 
\begin{equation}
\label{eq:schedulability-multiframe-rm-ratio-a}
(\frac{U_k}{f}+1) \prod_{j=1}^{k-1} (\frac{U_j}{f} + 1)\leq \frac{f+1}{f}
\end{equation}
or
\begin{equation}
  \label{eq:schedulability-multiframe-rm-ratio-b}
\sum_{j=1}^{k}U_j \leq  f k\left(\left({\frac{f+1}{f}}\right)^{\frac{1}{k}}-1\right),
\end{equation}
where $f$ is $\min_{i=1}^{k-1}
f_i$. Moreover, $f_i \geq \frac{\phi_i(1)}{\phi_i(2)-\phi_i(1)}$.
\end{theorem}
\begin{proof}
By Lemma~\ref{lemma:multi-frame-k-point-sufficient},
task $\tau_k$ is
schedulable by RM if there exists $t_j$ such that
\begin{equation}
  \phi_k(1)
  + \sum_{i=1}^{k-1} \alpha_i t_i U_i + \sum_{i=1}^{j-1} \beta_i t_i U_i \leq t_j,
\end{equation}
where, for a higher-priority task $\tau_i$, we know $\alpha_i =
\frac{\phi_i'(\frac{t_i}{T_i})}{\phi_i(1)\cdot \ell_i} \leq
1$ and $\beta_i =
\frac{\delta_i^{\min}(\ceiling{\frac{T_k}{T_i}})}{\phi_i(1)\cdot
  \frac{t_i}{T_i}} =
\frac{\delta_i^{\min}(\ell_i+1)}{\phi_i(1)\cdot \ell_i} = \frac{1}{f_i}$.  Then, suppose that
$\beta_i \leq \frac{1}{f} = \beta$. The rest of the proof is the same
as in the proof of Corollary \ref{col:implicit-general-ratio} to reach
the schedulability conditions in
Eqs.~\eqref{eq:schedulability-multiframe-rm-ratio-a} and
\eqref{eq:schedulability-multiframe-rm-ratio-b}.

Due to RM scheduling policy, we know that $\ell_i \geq 1$.
Based on the definition of function $\delta_i^{\min}(\ell)$, we know
that $\delta_i^{\min}(\ell)$ is a non-increasing function with respect
to $\ell$. Therefore, since $\ell_i \geq 1$, we know
that $\delta_i^{\min}(\ell_i+1) \leq
\delta_i^{\min}(2) = \phi_i(2)-\phi_i(1)$. With $\ell_i \geq 1$ and 
$\delta_i^{\min}(\ell_i+1) \leq  \phi_i(2)-\phi_i(1)$, we know that
$f_i
\geq \frac{\phi_i(1)}{\phi_i(2)-\phi_i(1)}$.
\end{proof}    

We can also have the following monotonic schedulability test:
\begin{theorem}
\label{thm:multiframe-fast}
An implicit-deadline multiframe system $\tau$ is schedulable under RM if
 \begin{equation}
\prod_{\tau_i \in \tau} (\frac{U_i}{f} + 1)\leq (1+\frac{1}{f})
      \end{equation}
where $f = min_{\tau_i \in \tau} \frac{\phi_i(1)}{\phi_i(2)-\phi_i(1)}$.
\end{theorem}

The hyperbolic bound test in
Eq.~\eqref{eq:schedulability-multiframe-rm-ratio-a} is the first one
for multi-frame tasks.  The result of the utilization bound in
Eq.~\eqref{eq:schedulability-multiframe-rm-ratio-b} is the same as the
result by Mok and Chen \cite{DBLP:dblp_journals/tse/MokC97} when $f$
is set to $\min_{i=1}^{k-1}
\frac{\phi_i(1)}{\phi_i(2)-\phi_i(1)}$.

\section*{Appendix D: Experiments}

This section presents evaluation results by measuring the \emph{success ratio} of the proposed tests with respect to a given goal of task set utilization. We will present the evaluations of utilization-based tests derived from our \framework{} framework and the existing tests for multiframe systems, DAG systems, and implicit-deadline systems. For each specified total utilization configuration, we generated 100 task sets.  The success ratio of a configuration is the number of task sets that are schedulable under RM (or global RM in DAG systems) divided by the number of task sets for this configuration, i.e., 100.

We first generated a set of sporadic tasks, and then the corresponding tasks were converted from this set according to different task models, e.g., multiframe and DAG tasks. 
The UUniFast method~\cite{bini2005measuring} was adopted to generate a set of utilization values with the given goal.
We here used the approach suggested by Davis and Burns~\cite{davis2008efficient} to generate the task periods according to an exponential distribution. 
The order of magnitude $p$ to control the period values between largest and smallest periods  is parameterized in evaluations. (E.g., $1-10ms$ for $p=1$, $1-100ms$ for $p=2$, etc.). 
The worst-case execution time was set accordingly, i.e., $C_{i,0}=T_iU_i$ for multiframe systems and $C_i=T_iU_i$ for DAG and uniprocessor implicit-deadline systems. Note that all the task systems are with implicit deadlines in our tests.

\subsection*{Evaluations for Multiframe}

The multiframe tasks were then converted from the sporadic tasks as follows:
The frame was generated in a similar manner to the method in~\cite{lu2007new}.
The size of frame types $m_i$ was randomly drawn from the the interval $[2,20]$.
For each frame we randomly chose a scaling factor $r_{i,j}$ in the range $(2,5)$ to assign its execution time based on that of the first frame, i.e. $C_{i,j}=C_{i,0}/r_{i,j}$.
The cardinality of the task set was 10.

\begin{figure}[t]
 \centering
    \includegraphics[width=\columnwidth]{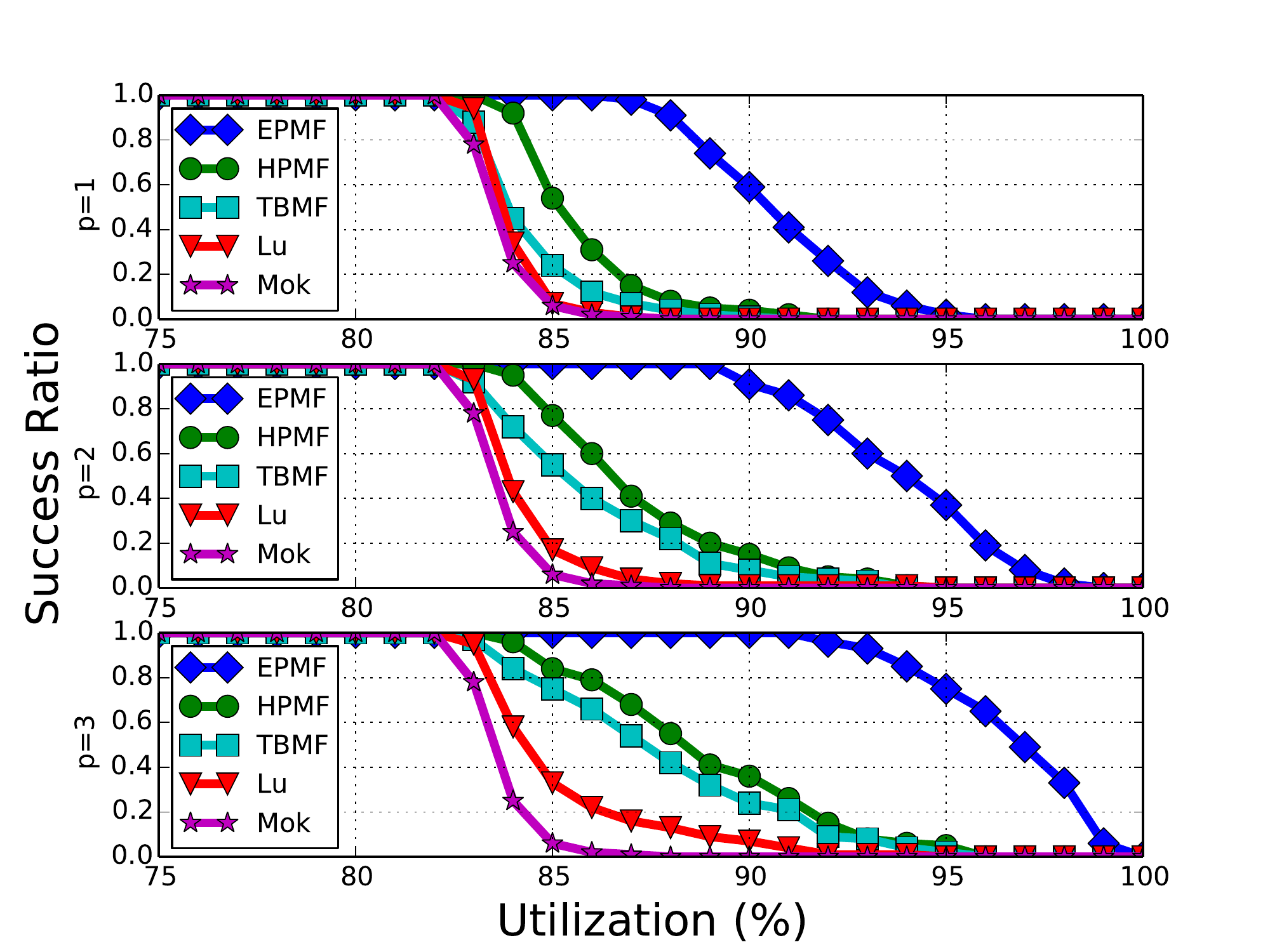}
    \caption{Success ratio comparison in multiframe systems}
    \label{fig:proportion-exp}
\end{figure}

In this experiment, the proposed tests (the first three) and the existing tests are listed as follows:
\vspace{-2mm}
 \begin{itemize}
  \item  Extreme Points Multiframe test (EPMF): by using Lemma~\ref{lemma:framework-general} in Theorem~\ref{thm:multiframe}.
   \item Hyperbolic Bound MultiFrame (HPMF): Eq.~\eqref{eq:schedulability-multiframe-rm-ratio-a}.
   \item Total utilization Bound MultiFrame (TBMF):  Eq.~\eqref{eq:schedulability-multiframe-rm-ratio-b}.
   \item Mok: Theorem 7 by Mok and Chen in \cite{DBLP:dblp_journals/tse/MokC97}.
  \item Lu: Theorem 3 by Lu et al. in \cite{lu2007new}. 
 \end{itemize}
\vspace{-2mm}

Figure~\ref{fig:proportion-exp} presents the result for the performance in terms of the success ratio.
For all tests, the success ratio are slightly better when the order of magnitude is greater. Our proposed tests are superior than the others for all different settings of $p$.

Note that the experiment conducted in~\cite{lu2007new} applies the technique of task merging proposed in~\cite{kuo2003efficient} as a preprocess and then tests the utilization bound.  Apparently, the former can be also used in our proposed tests.  However, we do not adopt this prepocess in our evaluations but focus on the effectiveness of utilization bounds themselves instead. The conclusion remains the same after adopting the preprocess on both sides.


\subsection*{Evaluations for DAG Task Systems}

Similarly, the DAG tasks were converted from the sporadic tasks as follows:
 The \emph{critical-path length} $\Psi_i$ of task $\tau_i$ was set by multiplying its WCET by uniform random 
values in the range $[0.75, 1]$. The following tests for global RM are evaluated:
\vspace{-2mm}
 \begin{itemize}
\item  Extreme Points DAG test (EPDAG): by using the following testing
$\frac{\Psi_k + \frac{C_k-\Psi_k}{M}}{T_k}  \leq 1 -  \sum_{i=1}^{k-1}  \frac{U_i(\alpha_i
  +\beta_i)}{\prod_{j=i}^{k-1} (\beta_jU_j + 1)}$ derived from Lemma~\ref{lemma:framework-general},
where $\alpha_i$ is defined as 
$\frac{2}{M}$
and $\beta_i$ as $\frac{1}{M}/\floor{\frac{T_k}{T_i}}$.
\item Hyperbolic Bound DAG (HPDAG): in Theorem~\ref{thm:multiprocessor-DAG}.
\item Chen and Agrawal Bound (CAB): in Corollary 4 in \cite{Chen+Agrawal2014}. CAB has the best known capacity augmentation bound for DAG systems under global RM.\footnote{The report in \cite{Chen+Agrawal2014} has more comprehensive bounds than its conference version in \cite{Li:ECRTS14}.}
 \end{itemize}
\vspace{-2mm}
The cardinality of the task set was 50.

Figure~\ref{fig:dag-exp} depicts results with different numbers of processors, i.e., $M=2,4,8$.
For these algorithms, the success ratios are better when the number of processors is less.
Apparently, our proposed tests are better than CAB, especially when the number of processors is large.
\begin{figure}[t]
 \centering
    \includegraphics[width=\columnwidth]{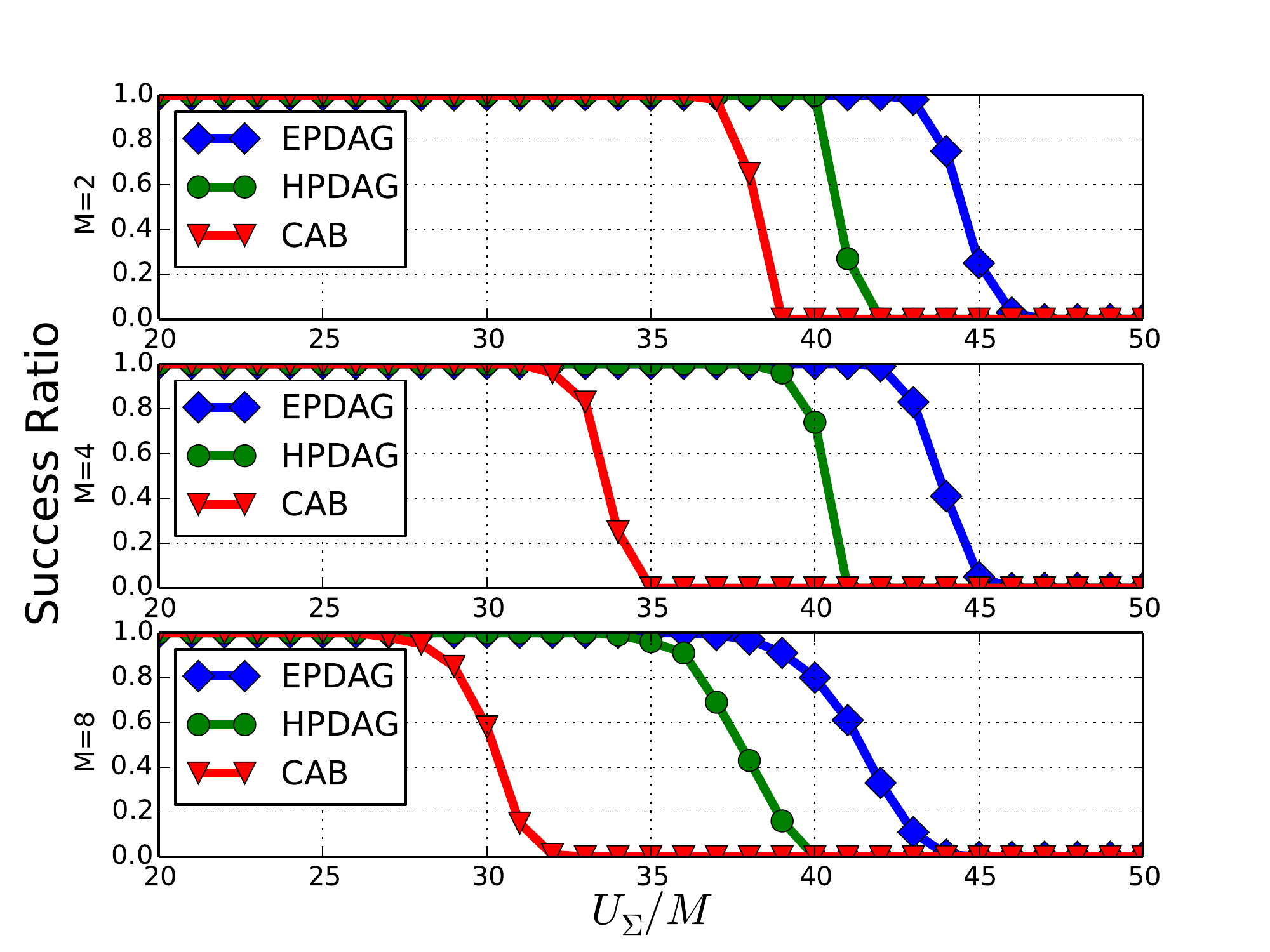}
    \caption{Success ratio comparison in DAG systems where $U_\Sigma$ is the total utilization of a DAG task set.}
    \label{fig:dag-exp}
\end{figure}
\subsection*{Evaluations for Implicit-Deadline Systems}
These tests in this experiment are as follows:
\vspace{-2mm}
  \begin{itemize}
 \item Extreme Points test (EP): We use Lemma~\ref{lemma:framework-general} in the analysis in Sec.~\ref{sec:application-constrained}.
 \item Hyperbolic Bound (HP): in Corollary~\ref{col:implicit-general-ratio}.
  \item Bini: in Corollary 2 in \cite{bini2009response}.
 \end{itemize}
\vspace{-2mm}
The cardinality of the task set was 10.

Figure~\ref{fig:exp-arbitrary} depicts the results for 3 different orders of magnitude, i.e., $p=1,2,3$.
For there tests, the success ratio is better if the order of magnitude is greater. The EP dominates all the other tests for all different orders of magnitude. On the other hand, the results from Bini et al. in \cite{bini2009response} and the proposed hyperbolic bound in 
Corollary~\ref{col:implicit-general-ratio} are comparable.
The performance by the proposed hyperbolic bound is better than by Bini et al. in \cite{bini2009response} for a smaller $p$ whereas Bini outperforms the proposed hyperbolic bound for a larger $p$. 

\begin{figure}[t]
 \centering
    \includegraphics[width=\columnwidth]{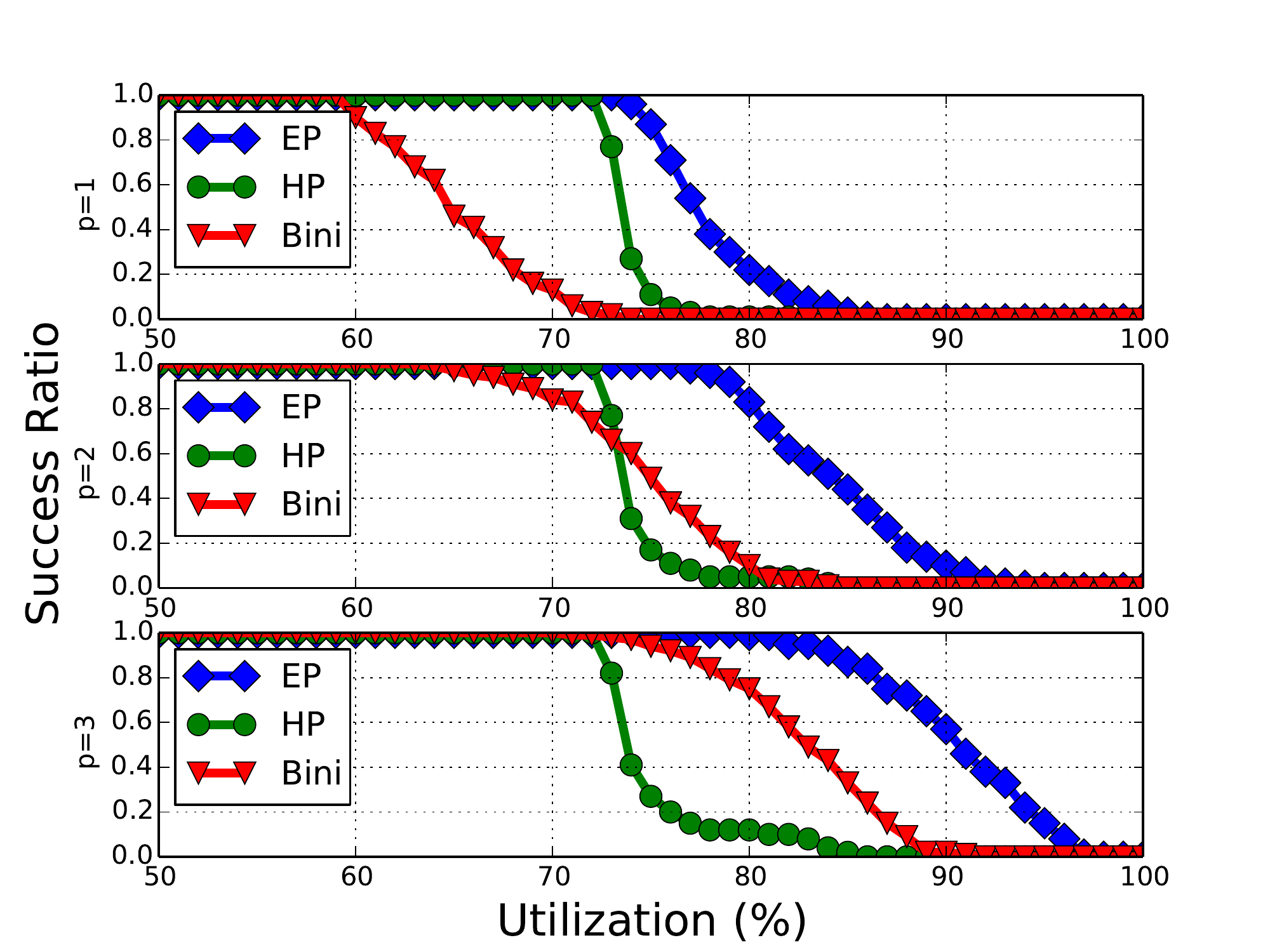}
    \caption{Success ratio comparison in implicit-deadline systems.}
    \label{fig:exp-arbitrary}
\end{figure}

\section*{Appendix E: Additional Properties in \framework{}}

We provide some additional properties that come directly from the \framework{} framework. These properties were not directly used in any of the above demonstrated examples. The first lemma is useful when the index of the $k-1$ higher priority tasks is not provided and cannot be determined while applying the schedulability tests. The results in Section~\ref{sec:framework} highly rely on the given order of the $k-1$ tasks. Therefore, without the given ordering, to be safe, we have to test all the permutations of the ordering of the $k-1$ tasks. Fortunately, the following lemma, as an extension of Lemma~\ref{lemma:framework-general}, shows that testing only one particular ordering is enough to provide a safe schedulability test.
\begin{lemma}
  \label{lemma:general-sorting}
  Suppose that the given $k$-point effective schedulability test, defined in Eq.~\eqref{eq:precodition-schedulability}, of a fixed-priority scheduling algorithm does not have a predefined order to index the $k-1$ higher-priority tasks.  Task $\tau_k$ is schedulable by the scheduling algorithm
 if
the following condition holds
\begin{equation}
\label{eq:schedulability-sorting}
0 < \frac{C_k}{t_k} \leq 1 -  \sum_{i=1}^{k-1}  \frac{U_i(\alpha_i
  +\beta_i)}{\prod_{j=i}^{k-1} (\beta_jU_j + 1)},
\end{equation}
by indexing the $k-1$ higher-priority tasks in a non-decreasing order of $\frac{\alpha_i}{\beta_i}$,
 in which  $0 < t_k$ and $0 < \alpha_i$ and $0 < \beta_i$ for any $i=1,2,\ldots,k-1$.
\end{lemma}
\begin{proof}
This lemma is proved by showing that the schedulability condition in Lemma~\ref{lemma:framework-general}, i.e., $1-\sum_{i=1}^{k-1}  \frac{U_i(\alpha_i
  +\beta_i)}{\prod_{j=i}^{k-1} (\beta_jU_j + 1)}$, is minimized, when the $k-1$ higher-priority tasks are indexed  in a non-decreasing order of $\frac{\alpha_i}{\beta_i}$. 
Suppose that there are two adjacent tasks $\tau_\ell$ and $\tau_{\ell+1}$ with $\frac{\alpha_\ell}{\beta_\ell} > \frac{\alpha_{\ell+1}}{\beta_{\ell+1}}$.
Let's now examine the difference of $\sum_{i=1}^{k-1}  \frac{U_i(\alpha_i
  +\beta_i)}{\prod_{j=i}^{k-1} (\beta_jU_j + 1)}$ by swapping the index of task $\tau_\ell$ and task $\tau_{\ell+1}$.  It can be easily observed that the other tasks $\tau_i$ with $i\neq \ell$ and $i\neq \ell+1$ do not change their corresponding values $\frac{U_i(\alpha_i
  +\beta_i)}{\prod_{j=i}^{k-1} (\beta_jU_j + 1)}$ in both orderings (before and after swapping $\tau_\ell$ and $\tau_{\ell+1}$). Suppose that $\frac{1}{\prod_{j=\ell}^{k-1} (\beta_jU_j + 1)}$ is $Q$, in which $Q > 0$. The difference in the term $\sum_{i=1}^{k-1}  \frac{U_i(\alpha_i
  +\beta_i)}{\prod_{j=i}^{k-1} (\beta_jU_j + 1)}$ before and after swapping tasks $\tau_\ell$ and $\tau_{\ell+1}$ is
  \begin{align*}\footnotesize
& \left( (\alpha_\ell + \beta_\ell)  U_\ell Q+ (\alpha_{\ell+1} + \beta_{\ell+1})  U_{\ell+1}  Q (1 + \beta_\ell U_\ell) \right) \\
& - \left( (\alpha_{\ell+1} + \beta_{\ell+1})  U_{\ell+1}  Q + (\alpha_\ell + \beta_\ell)  U_\ell  Q  (1 + \beta_{\ell+1}U_{\ell+1}) \right)\\
= & U_\ell   U_{\ell+1}   Q   (\beta_\ell \alpha_{\ell+1} - \beta_{\ell+1} \alpha_\ell) \\
 = &\beta_\ell\beta_{\ell+1}U_\ell   U_{\ell+1}   Q \left(\frac{\alpha_{\ell+1}}{\beta_{\ell+1}} - \frac{\alpha_\ell}{\beta_\ell}\right) <  0.
  \end{align*}\normalsize
Therefore, we reach the conclusion 
by repetitively swapping the tasks to achieve a  non-decreasing order of $\frac{\alpha_i}{\beta_i}$ for maximizing 
$\sum_{i=1}^{k-1}  \frac{U_i(\alpha_i
  +\beta_i)}{\prod_{j=i}^{k-1} (\beta_jU_j + 1)}$.
\end{proof}

Based on Lemma~\ref{lemma:general-sorting}, if all the $k-1$ higher-priority tasks have a constant $\beta_i=\beta > 0$, we know that they should be indexed  in a non-decreasing order of $\alpha_i$. There are also cases, in which we only know that the $k-1$ higher-priority tasks can be classified such that $h$ of them are associated with $\alpha_i = \alpha^\sharp$ and the other $k-1-h$ tasks are associated with $\alpha_i=\alpha^\flat$. This means that we are not sure whether task $\tau_i$ should be associated with $\alpha^\sharp$ or $\alpha^\flat$. Instead of testing all possible combinations, the following lemma shows that we only have to index the $k-1$ higher-priority tasks by their utilization non-decreasingly.

\begin{lemma}
  \label{lemma:general-fixed-beta-varied-alpha}
  Suppose that the $k$-point effective schedulability test, defined in Eq.~\eqref{eq:precodition-schedulability}, of a fixed-priority scheduling algorithm is defined with (1) a constant $\beta_i=\beta > 0$ for all the $k-1$ higher priority tasks and (2) uncertain $h$ higher-priority tasks ($h < k-1$) associated with $\alpha^\sharp$ and the remaining $k-1-h$ tasks with $\alpha^\flat$, where $\alpha^\sharp > \alpha^\flat > 0$. Task $\tau_k$ is schedulable by the scheduling algorithm
 if
the following condition holds
\begin{equation}
\label{eq:schedulability-sorting-varying-alpha}
0 < \frac{C_k}{t_k} \leq 1 -  \sum_{i=1}^{k-1}  \frac{U_i(\alpha_i
  +\beta)}{\prod_{j=i}^{k-1} (\beta U_j + 1)},
\end{equation}
by indexing the $k-1$ higher-priority tasks in a non-decreasing order of $U_i$ and assigning $\alpha_1, \alpha_2, \ldots, \alpha_{k-1-h}$ to $\alpha^\flat$ and $\alpha_{k-h}, \alpha_{k-h+1}, \ldots, \alpha_{k-1}$ to $\alpha^\sharp$.
\end{lemma}
\begin{proof}
  The assignment of the $\alpha_i$ values is due to Lemma~\ref{lemma:general-sorting}. Suppose $U_\ell > U_{\ell+1}$ for a certain $\ell$. There are three cases: (1) $\ell < k-1-h$, (2) $\ell = k-1-h$, and (3) $\ell > k-1-h$. For the first and the third cases, swapping the assignment of the utilization does not change the right-hand side of Eq.~\eqref{eq:schedulability-sorting-varying-alpha}. Only the second case matters. When $\ell = k-1-h$, it can be easily observed that due to the assumption $U_\ell > U_{\ell+1}$ and $\alpha_\ell < \alpha_{\ell+1}$, the right-hand side of Eq.~\eqref{eq:schedulability-sorting-varying-alpha} after swapping the assignment of the utilization (without swapping $\alpha_\ell$ and $\alpha_{\ell+1}$) is reduced. Therefore, the swapping makes the test harder. As a result, by repeating the above procedure, we reach the conclusion.
 \end{proof}
We will demonstrate how to use Lemma~\ref{lemma:general-fixed-beta-varied-alpha} in Appendix F for deriving a more precise hyperbolic bound with respect to global RM scheduling. The above lemma can be further generalized by allowing more levels of $\alpha_i$ values with a minor extension. However, there is no concrete clue whether such an extension may be useful. 
The following lemmas provide the utilization bound of Lemma~\ref{lemma:framework-totalU-constrained} for the case $\alpha+\beta > 1$.
\begin{lemma}
\label{lemma:framework-totalU-constrained-complete}
Suppose that $\alpha+\beta > 1$.
For a given $k$-point effective schedulability test of a scheduling
algorithm,  defined in
Definition~\ref{def:kpoints},
in which $0 < t_k$ and $0 < \alpha_i \leq \alpha$ and $0 < \beta_i \leq \beta$ for any
$i=1,2,\ldots,k-1$, task $\tau_k$ is schedulable by the scheduling
algorithm if 
\begin{equation}\small 
\label{eq:schedulability-totalU-constrained-complete}
\frac{C_k}{t_k} + \sum_{i=1}^{k-1}U_i \leq 
\begin{cases}
  \frac{\ln (1+\frac{\beta}{\alpha})}{\beta}, &  (\alpha+\beta)^{\frac{1}{k}} < \alpha\\  
  \frac{\ln (\alpha+\beta)-\alpha+1}{\beta}, & \mbox{otherwise.}\\  
\end{cases}
\end{equation}
\end{lemma}
\begin{proof}
  With the assumption $\alpha+\beta > 1$, we only have to consider the last two cases in Eq.~\eqref{eq:schedulability-totalU-constrained}. In both cases, the minimum bound happens when $k$ goes to $\infty$. For the case with $\frac{(k-1)((1+\frac{\beta}{\alpha})^{\frac{1}{k-1}}-1)}{\beta}$,  this is lower bounded by $\frac{\ln (1+\frac{\beta}{\alpha})}{\beta}$ when $k \rightarrow \infty$. For the other case, we have $\frac{(k-1)((\alpha+\beta)^{\frac{1}{k}}-1)+((\alpha+\beta)^{\frac{1}{k}}-\alpha)}{\beta} =  \frac{k((\alpha+\beta)^{\frac{1}{k}}-1)+1-\alpha}{\beta}$.
\end{proof}

\begin{lemma}
\label{lemma:framework-totalU-constrained-complete2}
Suppose that $\alpha+\beta > 1$.
For a given $k$-point effective schedulability test of a scheduling
algorithm,  defined in
Definition~\ref{def:kpoints},
in which $0 < t_k$ and $0 < \alpha_i \leq \alpha$ and $0 < \beta_i \leq \beta$ for any
$i=1,2,\ldots,k-1$, task $\tau_k$ is schedulable by the scheduling
algorithm if 
\begin{equation}\small 
\label{eq:schedulability-totalU-constrained-complete2}
\frac{1}{\alpha}\frac{C_k}{t_k} + \sum_{i=1}^{k-1}U_i \leq\frac{ k((1+\frac{\beta}{\alpha})^{\frac{1}{k}}-1)}{\beta}, 
\end{equation}
in which the right-hand side of Eq.~\eqref{eq:schedulability-totalU-constrained-complete2} is lower bounded by  $\frac{\ln (1+\frac{\beta}{\alpha})}{\beta}$ when $k\rightarrow \infty$.
\end{lemma}
\begin{proof}
  This comes directly from Eq.~\eqref{eq:schedulability-constrained}
  in Lemma~\ref{lemma:framework-constrained} with a simpler Lagrange
  Multiplier procedure. That is, the infimum $\frac{1}{\alpha}\frac{C_k}{t_k} + \sum_{i=1}^{k-1}U_i$ with $(\frac{C_k}{t_k}+\frac{\alpha}{\beta})(\beta U_1+1)^{k-1} > \frac{\alpha}{\beta}+1$ happens when $\frac{\beta}{\alpha}\frac{C_k}{t_k}= \beta U_1 = \beta U_2 = \cdots = \beta U_{k-1} = ((1+\frac{\beta}{\alpha})^{\frac{1}{k}}-1)$. 
\end{proof}
With Lemma~\ref{lemma:framework-totalU-constrained-complete2}, the first case in Eq.~\eqref{eq:schedulability-totalU-constrained-complete} already pessimistically holds when $\alpha \geq 1$. The above two lemmas generalize the results in \cite{RTSS14a} for the case $1 \leq \alpha \leq 2$ and $\beta = 1$. 

\vspace{-0.13in}
\section*{Appendix F: Improved Global RM Test}

The test in Section~\ref{sec:multiprocessor} of global RM scheduling for sporadic tasks can be improved by using a tighter schedulability test from Guan et al. \cite{DBLP:conf/rtss/GuanSYY09}. It has been concluded by Guan et al. \cite{DBLP:conf/rtss/GuanSYY09} that we only have to consider $M-1$ tasks with carry-in jobs, for constrained-deadline task sets, when considering task $\tau_k$ with $k > M$. For implicit-deadline task sets, this means that we only need to set $\alpha_i$ of some tasks to $\frac{2}{M}$, rather than all the $k-1$ tasks in Eq.~\eqref{eq:W_i-multiprocessor}. More precisely, we can define two different time-demand functions, depending on whether task $\tau_i$ is with a carry-in job or not:\footnote{We still use the step-wise function here, which is an over-approximation of the linear function used by Guan et al. \cite{DBLP:conf/rtss/GuanSYY09}. The procedure here is a bit different as we only present a concept to transform the test to the \framework{} framework. In \cite{DBLP:conf/rtss/GuanSYY09}, the analysis for constrained-deadline task sets first requires to stretch the window of interest with a length $\phi$. However, such a length $\phi$ should be $0$ in the worst case \cite{DBLP:conf/rtss/GuanSYY09}. Therefore, we are more pessimistic here.}
\begin{equation}
  \label{eq:W_i-carryin}
W_i^{carry}(t) =
\begin{cases}
  C_i & 0 < t < C_i\\
  C_i + \ceiling{\frac{t-C_i}{T_i}}C_i & otherwise,
\end{cases}
\end{equation}
and
\begin{equation}
  \label{eq:W_i-normal}
W_i^{normal}(t) = \ceiling{\frac{t}{T_i}}C_i.
\end{equation}
Moreover, we can further over-approximate $W_i^{carry}(t)$, since  $W_i^{carry}(t) \leq W_i^{normal}(t)+C_i$. Therefore, a sufficient schedulability test for testing task $\tau_k$ with $k > M$ for global RM is to verify whether 
\begin{equation}
  \label{eq:grm-multiprocessor-M-1-carryin}
\exists 0 < t \leq T_k, C_k + \frac{(\sum_{\tau_i \in {\bf T}'} C_i) +  (\sum_{i=1}^{k-1}W_i^{normal}(t)) }{M} \leq t.  
\end{equation}
for all ${\bf T}' \subseteq hp(\tau_k)$ with $|{\bf T}'| = M-1$. That is, if
the above test in Eq.~\eqref{eq:grm-multiprocessor-M-1-carryin} passes for any selection of $M-1$ higher-priority tasks to have carry-in jobs, then task $\tau_k$ is schedulable by global RM.  For a given carry-in task set ${\bf T}'$ with $|{\bf T}'|=M-1$, the translation to the \framework{} framework is as follows: (1) the parameters are $0 < \alpha_i \leq \frac{2}{M}$ and $0 < \beta_i \leq \frac{1}{M}$ by using Eq.~\eqref{eq:grm-multiprocessor-M-1-carryin} if $\tau_i$ is in ${\bf T}'$, and (2) 
the parameters are $\alpha_i = \frac{1}{M}$ and $0 < \beta_i \leq \frac{1}{M}$ by using Eq.~\eqref{eq:grm-multiprocessor-M-1-carryin} if $\tau_i$ is not in ${\bf T}'$. Therefore, this is exactly the case when Lemma~\ref{lemma:general-fixed-beta-varied-alpha} can be adopted by selecting $M-1$ higher-priority tasks with $\alpha_i=\frac{2}{M}$ and the remaining $k-M$ higher-priority tasks with $\alpha_i=\frac{1}{M}$, where $\beta_i$ is set safely to $\frac{1}{M}$. Therefore, we reach the conclusion with the following theorem by the above analysis and Lemma~\ref{lemma:general-fixed-beta-varied-alpha}. 

\begin{theorem}
\label{thm:multiprocessor-GRM-M-1-carry}
Task $\tau_k$ with $k > M$ in a sporadic implicit-deadline task system is
schedulable by global RM on $M$ processors if
\begin{equation}
\label{eq:schedulability-GRM-M-1-carry}
0 < U_k\leq 1 -  \sum_{i=1}^{k-1}  \frac{U_i(\alpha_i
  +\frac{1}{M})}{\prod_{j=i}^{k-1} (\frac{1}{M} U_j + 1)},
\end{equation}
by indexing the $k-1$ higher-priority tasks in a non-decreasing order of $U_i$ and assigning $\alpha_1, \alpha_2, \ldots, \alpha_{k-M}$ to $\frac{1}{M}$ and $\alpha_{k-M+1}, \alpha_{k-M+2}, \ldots, \alpha_{k-1}$ to $\frac{2}{M}$.
\end{theorem}

\end{document}